\documentclass[prodmode,acmtecs]{acmsmall} 

\usepackage[linesnumbered, ruled]{algorithm2e}
\renewcommand{\algorithmcfname}{ALGORITHM}
\SetAlFnt{\small}
\SetAlCapFnt{\small}
\SetAlCapNameFnt{\small}
\SetAlCapHSkip{0pt}
\IncMargin{-\parindent}

\acmVolume{9}
\acmNumber{4}
\acmArticle{39}
\acmYear{2010}
\acmMonth{3}

\newtheorem{theo}{Theorem}
\newtheorem{lemm}{Lemma}
\newtheorem{defn}{Definition}

\newcommand{\task}[1] {\tau_{#1}}
\newcommand{\proc}[1] {\tau_{#1}^{proc}}
\newcommand{\pri}[1] {\tau_{#1}^{pri}}
\newcommand{\BCET}[1] {\tau_{#1}^{BCET}}
\newcommand{\WCET}[1] {\tau_{#1}^{WCET}}
\newcommand{\minR}[1] {\tau_{#1}^{minR}}
\newcommand{\maxR}[1] {\tau_{#1}^{maxR}}
\newcommand{\minS}[1] {\tau_{#1}^{minS}}
\newcommand{\maxS}[1] {\tau_{#1}^{maxS}}
\newcommand{\minF}[1] {\tau_{#1}^{minF}}
\newcommand{\maxF}[1] {\tau_{#1}^{maxF}}
\newcommand{\release}[1] {\tau_{#1}^{release}}
\newcommand{\start}[1] {\tau_{#1}^{start}}
\newcommand{\finish}[1] {\tau_{#1}^{finish}}
\newcommand{\graph}[1] {\mathcal{T}_{#1}}
\newcommand{\period}[1] {\mathcal{T}_{#1}^{p}}
\newcommand{\jitter}[1] {\mathcal{T}_{#1}^{j}}
\newcommand{\graphRelease}[1] {\mathcal{T}_{#1}^{release}}
\newcommand{\deadline}[1] {\mathcal{T}_{#1}^{d}}
\newcommand{\rphase}[2] {\phi_{#1,#2}^r}
\newcommand{\sphase}[2] {\phi_{#1,#2}^s}
\newcommand{\fphase}[2] {\phi_{#1,#2}^f}
\newcommand{\pshift}[2] {\Psi_{#1,#2}}
\newcommand{\condset}[2] {\mathcal{#1}_{\task{#2}}}
\newcommand{\ignore}[1]{}

\begin{document}

\markboth{J. Choi et al.}{A Hybrid Performance Analysis Technique for Distributed Real-Time Embedded Systems}

\title{A Hybrid Performance Analysis Technique for Distributed Real-Time Embedded Systems}
\author{
JUNCHUL CHOI \affil{Seoul National University}
HYUNOK OH \affil{Hanyang University}
SOONHOI HA \affil{Seoul National University}
}

\begin{abstract}
It remains a challenging problem to tightly estimate the worst case response time of an application in a distributed embedded system, especially when there are dependencies between tasks. We discovered that the state-of-the art techniques considering task dependencies either fail to obtain a conservative bound or produce a loose upper bound. We propose a novel conservative performance analysis, called hybrid performance analysis, combining the response time analysis technique and the scheduling time bound analysis technique to compute a tighter bound fast. Through extensive experiments with randomly generated graphs, superior performance of our proposed approach compared with previous methods is confirmed.
\end{abstract}




\acmformat{Junchul Choi, Hyunok Oh, and Soonhoi Ha, 2015. A Hybrid Performance Analysis Technique for Distributed Real-Time Embedded Systems.}


\begin{bottomstuff}

Author's addresses: 
J. Choi, Department of Computer Science and Engineering, Seoul National University;
H. Oh, Department of Information Systems, Hanyang University;
S. Ha, Department of Computer Science and Engineering, Seoul National University.
\end{bottomstuff}

\maketitle

\section{Introduction}

For the design of embedded systems that support real-time applications, it is required to guarantee the satisfaction of real-time constraints. After applications are mapped to a candidate architecture, we check the feasibility of the architecture by estimating the performance. Fast estimation enables us to explore the wider design space of architecture selection and application mapping. More accurate estimation will reduce the system cost. The performance analysis problem addressed in this paper is to estimate the worst case response time (WCRT) of an application that is executed on a distributed embedded system. A good example can be found in an intelligent safety application in a car where there is a tight requirement on the worst case response time from the sensor input to the actuator output.

Despite a long history of research over two decades, it still remains a challenging problem to tightly estimate the WCRT of an application in a distributed embedded system based on a fixed priority scheduling policy. Since the response time of an application is affected by interference between applications as well as execution time variation of tasks, all possible execution scenarios should be considered to obtain the exact WCRT. There are some approaches proposed, such as a model checking approach [Brekling et al. 2008]\ignore{\cite{ref1}} and an ILP-based approach [Kim et al. 2012]\ignore{\cite{ref2}}, to find the accurate WCRT. However, they require exponential time complexity. The exact WCRT analysis problem is known to be NP-complete [Yen and Wolf 1998]\ignore{\cite{ref3}}.

Analytical techniques have been extensively researched to obtain a tight upper bound of the WCRT with diverse assumptions on target architectures and applications. This paper assumes that an application is given as a task graph that represents data dependency between tasks and the execution time of a task may vary. It is assumed that each task has a fixed priority. In addition, we support arbitrary mixture of preemptive and non-preemptive processing elements in the system. To analyze the WCRT of an application, this paper proposes a performance analysis technique, called hybrid performance analysis (HPA), combining a scheduling time bound analysis and a response time analysis (RTA). The proposed technique is proven to be conservative and experimental results show that it provides a tighter bound of WCRT than the other state-of-the-art techniques.

The rest of this paper is organized as follows. In Section \ref{sec:related} we overview the related work and highlight the contributions of this work. In Section \ref{sec:problem_def}, the application model and the system model assumed in this paper are formally described. Section \ref{sec:YW_review} reviews the Y\&W method and proves that it is not conservative by showing some counter examples. The proposed technique and some optimization techniques are explained in Sections \ref{sec:HPA} and \ref{sec:enhanced_HPA} respectively. We summarize the overall algorithm in Section \ref{sec:overall_alg}. Experimental results are discussed in Section \ref{sec:experiment}. Finally, we conclude this paper in Section \ref{sec:conclusion}.

\section{Related Work} \label{sec:related}

Response time analysis (RTA) was first introduced for a single processor system based on preemptive scheduling of independent tasks that have fixed priorities, fixed execution times, and relative deadline constraints equal to their periods [Lehoczky et al. 1989]\ignore{\cite{ref4}}. Extensive research efforts [Lehoczky 1990, Audsley et al. 1993]\ignore{\cite{ref6}\cite{ref7}} have been performed to release the restricted assumptions. Pioneered by K. Tindell et al. [Tindell and Clark 1994]\ignore{\cite{ref8}}, a group of researchers extended the schedulability analysis technique to distributed systems; for example, dynamic offset of tasks [Palencia and Harbour 1998]\ignore{\cite{ref9}}, communication scheduling [Tindel et all. 1995]\ignore{\cite{ref10}}, partitioned scheduling with shared resources [Schliecker and Ernst 2010]\ignore{\cite{ref11}}, and earliest deadline first (EDF) scheduling [Pellizzoni and Lipari 2007]\ignore{\cite{ref12}}. There exist some researches that consider precedence constraints between tasks, by assigning the offset and deadline of each task conservatively considering every possible execution ordering between tasks [Palencia and Harbour 1998, Tindel et all. 1995, Pellizzoni and Lipari 2007], which usually incur significant overhead of overestimation and computation complexity. On the other hand, we handle the dependent tasks directly, assuming that a task is released immediately after all predecessors complete.

To the best of our knowledge, the state of the art RTA method for dependent tasks was proposed by Yen and Wolf [Yen and Wolf 1998], denoted as the Y\&W method hereafter. It considers data dependencies between tasks and variable execution times of tasks directly in the response time analysis but supports only preemptive systems. Several extensions have been made to the Y\&W method, with considering communication [Yen and Wolf 1995]\ignore{\cite{ref13}} and control dependencies [Pol et al. 2000]\ignore{\cite{ref14}}. Unfortunately, the conservativeness has not been proven for the Y\&W method. This paper discovers that the Y\&W method is actually not conservative by showing counter-examples for which the Y\&W method produces a shorter response time than the worst case response time in Section \ref{sec:YW_review}.

Non-preemptive processing elements are supported in the MAST suite [Harbour 2001]\ignore{\cite{ref16}} that includes several schedulability analysis techniques. But they support only chain-structured graphs where a task has a single input and/or a single output port.

There is a compositional approach distinguished from the holistic RTA-based approaches. SymTA/S [Henia et al. 2005]\ignore{\cite{ref15}} which is a well-known compositional analysis, performs the analysis in a modular manner. It analyzes the performance for each processing component and abstracts its result as an event stream at the component boundary. While the compositional approach achieves scalability, it sacrifices estimation accuracy by ignoring the release time constraint coming from data dependencies between tasks running on different processing elements.

Recently, a holistic WCRT analysis approach, called scheduling time bound analysis (STBA), has been proposed [Kim et al. 2013]\ignore{\cite{ref5}}. It computes the conservative time bound for each task within which the task will be scheduled, considering all possible scheduling patterns. In the STBA approach, however, the task graphs should be expanded up to the least common multiple (LCM) of their periods, which limits the scalability of the technique. While the proposed technique adopts the basic time bound idea of the STBA method to analyze how data dependencies affect the release times of tasks in the same application, it considers inter-application interference analytically, based on the response time analysis.

\section{Problem Definition} \label{sec:problem_def}

\begin{figure}[ht]
\centerline{\includegraphics[width=11cm]{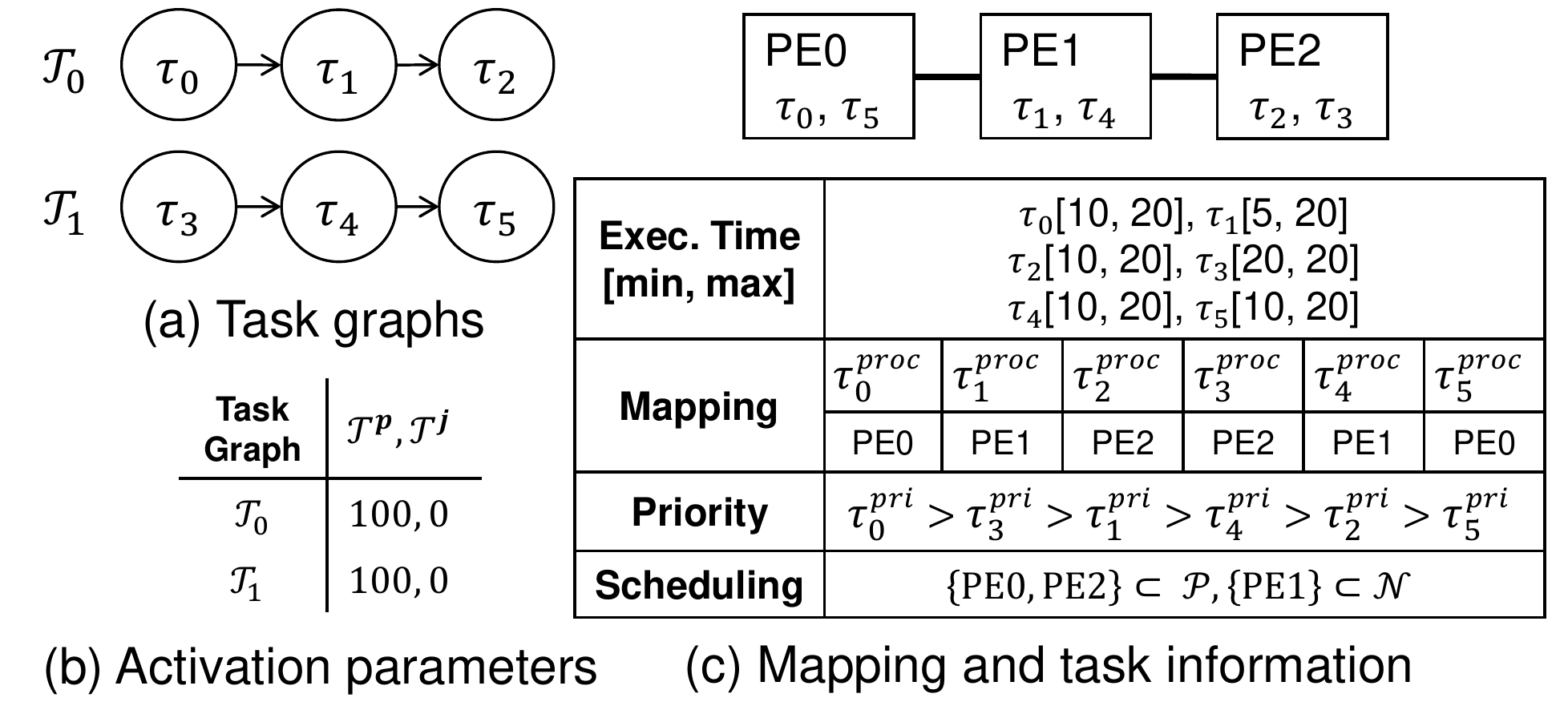}}
\caption{(a) Example task graphs, (b) task graph information, (c) task mapping and task information}
\label{fig:Problem}
\end{figure}

We formally describe the application model and the system model assumed in this paper. An input application, $\graph{i}$, is represented as an acyclic task graph as illustrated in Fig. \ref{fig:Problem} (a). In a task graph, $\graph{}=\{\mathcal{V},\mathcal{E}\}$, $\mathcal{V}$ represents a set of tasks and $\mathcal{E}=\{(\task{1}, \task{2})|\task{1}, \task{2}\in\mathcal{V}\}$ a set of edges to represent execution dependencies between tasks. If a task has more than one input edge, it is released after all predecessor tasks are completed. An application $\graph{}$ can be initiated periodically or sporadically, characterized by a tuple $(\period{}, \jitter{})$ where $\period{}$ and $\jitter{}$ represent the period and the maximum jitter, respectively. For sporadic activation, $\period{}$ denotes the minimum initiation interval. Task graph $\graph{}$ is given a relative deadline $\deadline{}$ to meet once activated. We assume that $\deadline{}$ is not greater than $\period{}$ in this paper. The task graph that task $\task{i}$ belongs to is denoted by $\graph{\task{i}}$.

A system consists of a set of processing elements (PEs) as shown in Fig. \ref{fig:Problem} (c). A task is a basic mapping unit onto a processing element. We assume that task mapping is given and fixed. The processing element that the task $\task{}$ is mapped to is denoted by $\proc{}$. For each task $\task{}$, the varying execution time is represented as a tuple $[\BCET{}, \WCET{}]$ indicating the lower and the upper bound on the mapped PE. Note that a communication network can be modeled as a separate PE. For instance, the PE graph of Fig. \ref{fig:Problem} (c) represents a system that consists of two processors (PE0 and PE2) connected to a bus (PE1). Tasks mapped to a communication network deliver messages between two computation tasks; for example $\task{1}$ indicates message communication between two computation tasks, $\task{0}$ and $\task{2}$.

We assume that the scheduling policy of a PE can be either a fixed-priority preemptive scheduling or a fixed-priority non-preemptive scheduling. $\mathcal{P}$ denotes a set of PEs that have preemptive scheduling policy, and $\mathcal{N}$ denotes a set of PEs with non-preemptive scheduling policy. A PE belongs to either $\mathcal{P}$ or $\mathcal{N}$. In Fig. \ref{fig:Problem}, PE0 and PE2 use preemptive scheduling while PE1 serves the communication tasks in a non-preemptive fashion; a higher-priority message cannot preempt the current message delivery. We assume that all tasks mapped to each PE have distinct priorities to make the scheduling order deterministic. The priority of the task $\task{}$ is denoted by $\pri{}$.

The WCRT of task graph $\graph{}$, denoted by $\mathcal{R}_{\graph{}}$, is defined as the time difference between the latest finish time and the earliest release time among tasks in the task graph.

\section{Review of the Y\&W Method} \label{sec:YW_review}

Since the Y\&W method is known as a state-of-the art schedulability analysis technique that considers dependency between tasks directly, we select it as the reference technique for comparison in this paper. In this section we review the key ideas of the Y\&W method and prove that it fails to find a conservative upper bound of the WCRT. Since dependency between tasks constrains the release times of tasks, the Y\&W method proposed three techniques: \emph{separation analysis}, \emph{phase adjustment}, and \emph{period shifting}.

\begin{figure}[ht]
\centerline{\includegraphics[width=9cm]{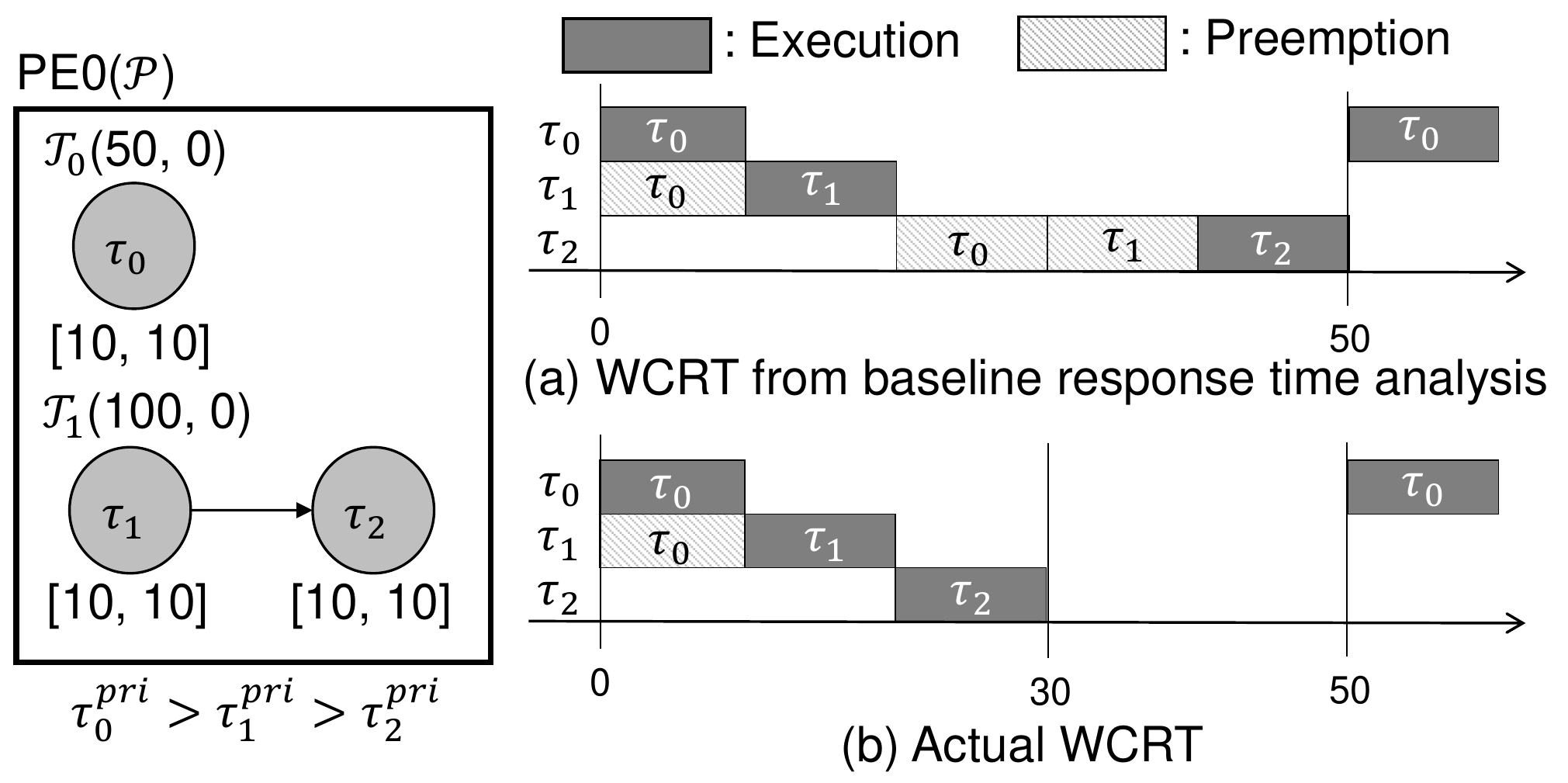}}
\caption{A simple example of WCRT estimation}
\label{fig:Ex_baseRTA}
\end{figure}

The baseline response time analysis for independent tasks defines the response time $r_i$ of a task $\task{i}$ as an iterative form:

\begin{equation} \label{eq:Eq1}
\centering
r_i=\WCET{i}+\sum_{\task{j}\in\{\task{j}|\pri{i}<\pri{j}\}}{\bigg\lceil\frac{r_i}{\period{\task{j}}}\bigg\rceil\cdot\WCET{j}}
\end{equation}

where set ${\task{j}\in\{\task{j}|\pri{i}<\pri{j}\}}$ represents the set of higher priority tasks that may preempt task $\task{i}$. Note that the second term subsumes all possible preemption delays by higher priority tasks for conservative estimation. When the baseline RTA is directly applied to tasks with dependency, the analytical WCRT of a task graph may become much larger than the actual WCRT. Consider an example in Fig. \ref{fig:Ex_baseRTA}. Unless the task dependency is considered, $\task{0}$ can preempt both $\task{1}$ and $\task{2}$, and $\task{1}$ can preempt $\task{2}$. Therefore, WCRT becomes $r_1+r_2=50$ as shown in Fig. \ref{fig:Ex_baseRTA} (a) while the tight WCRT is $30$. If the dependency is considered, $\task{0}$ can preempt either $\task{1}$ or $\task{2}$ but not both, and $\task{1}$ cannot preempt $\task{2}$, as shown in Fig. \ref{fig:Ex_baseRTA} (b).

In the Y\&W method, the release times and the finish times of tasks are computed, considering data dependencies between tasks. If the earliest release time ($\minR{j}$) of a higher priority task $\task{j}$ is larger than the latest finish time ($\maxF{i}$) of a given task $\task{i}$, $\task{j}$ cannot preempt $\task{i}$. And if there is a direct path between a higher priority task and a given task, no preemption will occur. These cases are identified in the \emph{separation analysis} to check the preemption possibility between two tasks in the same task graph; the \emph{separation analysis} finds that $\task{1}$ cannot preempt $\task{2}$ in the example of Fig. \ref{fig:Ex_baseRTA}.

The \emph{phase adjustment} technique in the Y\&W method computes the distance, called phase, between the release times of the preempting and the preempted tasks to identify the impossible preemptions along dependent tasks. The request phase $\rphase{i}{j}$ means the minimum distance from the request(release) time of $\task{i}$ to the next release time of $\task{j}$. Note that if the dependency is not considered, the request phase will be 0 since the worst case preemption occurs when the preempting task is released at the same time as the preempted task. The request phase is dependent on the finish phase of its predecessors. The finish phase $\fphase{i}{j}$ means the minimum distance from the finish time of $\task{i}$ to the next release time of $\task{j}$. The \emph{phase adjustment} technique modifies the response time formula of equation (\ref{eq:Eq1}) as follows:

\begin{equation} \label{eq:Eq2}
\centering
\rphase{i}{j}=\max\big(0,\min_{\task{k}\in{pred(\task{i})}}{(\fphase{k}{j}+\maxF{k})}-\max_{\task{k}\in{pred(\task{i})}}{\maxF{k}}\big)
\end{equation}

\begin{equation} \label{eq:Eq3}
\centering
r_i=\WCET{i}+\sum_{\task{j}\in{\{\task{j}|\pri{i}<\pri{j}\}-separated[\task{i}]}}{\bigg\lceil\frac{\max(0, r_i-\rphase{i}{j})}{\period{\task{j}}}\bigg\rceil\cdot\WCET{j}}
\end{equation}

\begin{equation} \label{eq:Eq4}
\centering
\fphase{i}{j}=\left\{\begin{array}{rl}
	\max(0,\rphase{i}{j}-r_i), & \textrm{if $\task{j}\not\in{\{\task{j}|\pri{i}<\pri{j}\}-separated[\task{i}]}$}\\
	(\rphase{i}{j}-r_i) \bmod \period{j} & \textrm{otherwise}\\
\end{array} \right.
\end{equation}

where $separated[\task{i}]$ is a set of higher priority tasks excluded in the preemption delay computation by the \emph{separation analysis}, and $pred(\task{i})$ is the immediate predecessor set of task $\task{i}$. In the modified formula, the request phase is subtracted from the response time in the preemption count computation. For conservative computation, the request phase and the finish phase are made non-negative. For detailed explanation of the formula, refer to [Yen and Wolf 1998]\ignore{\cite{ref3}}.

For a preempting task $\task{j}$, we compute the request phase and the finish phase for each task $\task{i}$ of a task graph. For a source task $\task{i}$, $\rphase{i}{j}$ is set to 0, meaning that the preempting task can be released at the same time to make $\task{i}$ experience the maximum number of preemptions by $\task{j}$. After computing the response time, finish phase $\fphase{i}{j}$ is updated. The request phase for a non-source task is updated in turn based on the finish phases of its predecessors. In the example of Fig. \ref{fig:Ex_baseRTA}, the request phase $\rphase{1}{0}$ is initialized to 0 and the response time of $\task{1}$ is 20. The finish phase $\fphase{1}{0}$ is updated to $(0-20)\bmod50=30$, and the request phase $\rphase{2}{0}$ is inherited from $\fphase{1}{0}$. Since $\task{1}$ and $\task{2}$ are separated, the response time of $\task{2}$ becomes $10=10+\lceil{\frac{max⁡(0,10-30)}{50}}\rceil\cdot10$. If we sum the response times of $\task{1}$ and $\task{2}$, we obtain the actual WCRT that is 30.

Consider a preempting task with varying release time. The release time may vary due to finish time variation of its predecessor. Then, the preempting task will not be scheduled periodically, but in a bursty fashion, which increases the preemption count. The \emph{period shifting} technique is used in the Y\&W method to account for this effect. Let $\task{j}$ be a task that may preempt $\task{i}$ and of which release time varies between $\minR{j}$ and $\maxR{j}$ that indicate the minimum and maximum release time of $\task{j}$ respectively. To compute the maximum number of preemptions by $\task{j}$ onto $\task{i}$, we need to consider the release time difference.

The authors of the Y\&W method did not clarify how to integrate the \emph{period shifting} technique and the \emph{phase adjustment} technique in a single formula. If both techniques are applied separately in sequence, the conventional RTA equation is modified to the following formula, which is used as the Y\&W method in this paper.

\begin{equation} \label{eq:Eq5}
\centering
r_i=\WCET{i}+\sum_{\task{j}\in{\{\task{j}|\pri{i}<\pri{j}\}-separated[\task{i}]}}{\bigg\lceil\frac{\max(0, r_i-\rphase{i}{j}+\maxR{j}-\minR{j})}{\period{\task{j}}}\bigg\rceil\cdot\WCET{j}}
\end{equation}

\subsection{The Y\&W method is NOT conservative}

Unfortunately, the Y\&W method fails to find the maximum preemption count. Consider a simple example of Fig. \ref{fig:Ex_underestimation1}. Since $\task{2}$ has no predecessor and a zero jitter, release time is always 0. Then the response time of $\task{1}$ from the Y\&W method becomes $25=20+\lceil{\frac{(25+0-0)}{30}}\rceil\cdot5$ as shown in Fig. \ref{fig:Ex_underestimation1} (a), in which $\task{1}$ can be preempted once by $\task{2}$. But $\task{2}$ can preempt $\task{1}$ twice and the response time of $\task{1}$ can be as large as 30 as shown in Fig. \ref{fig:Ex_underestimation1} (b) when the start of $\task{2}$ is delayed by the preemption of $\task{0}$. Since the Y\&W method only considers the release time variation by predecessors but not the start time variation by preemption, it fails to obtain a conservative WCRT.

\begin{figure}[ht]
\centerline{\includegraphics[width=9.5cm]{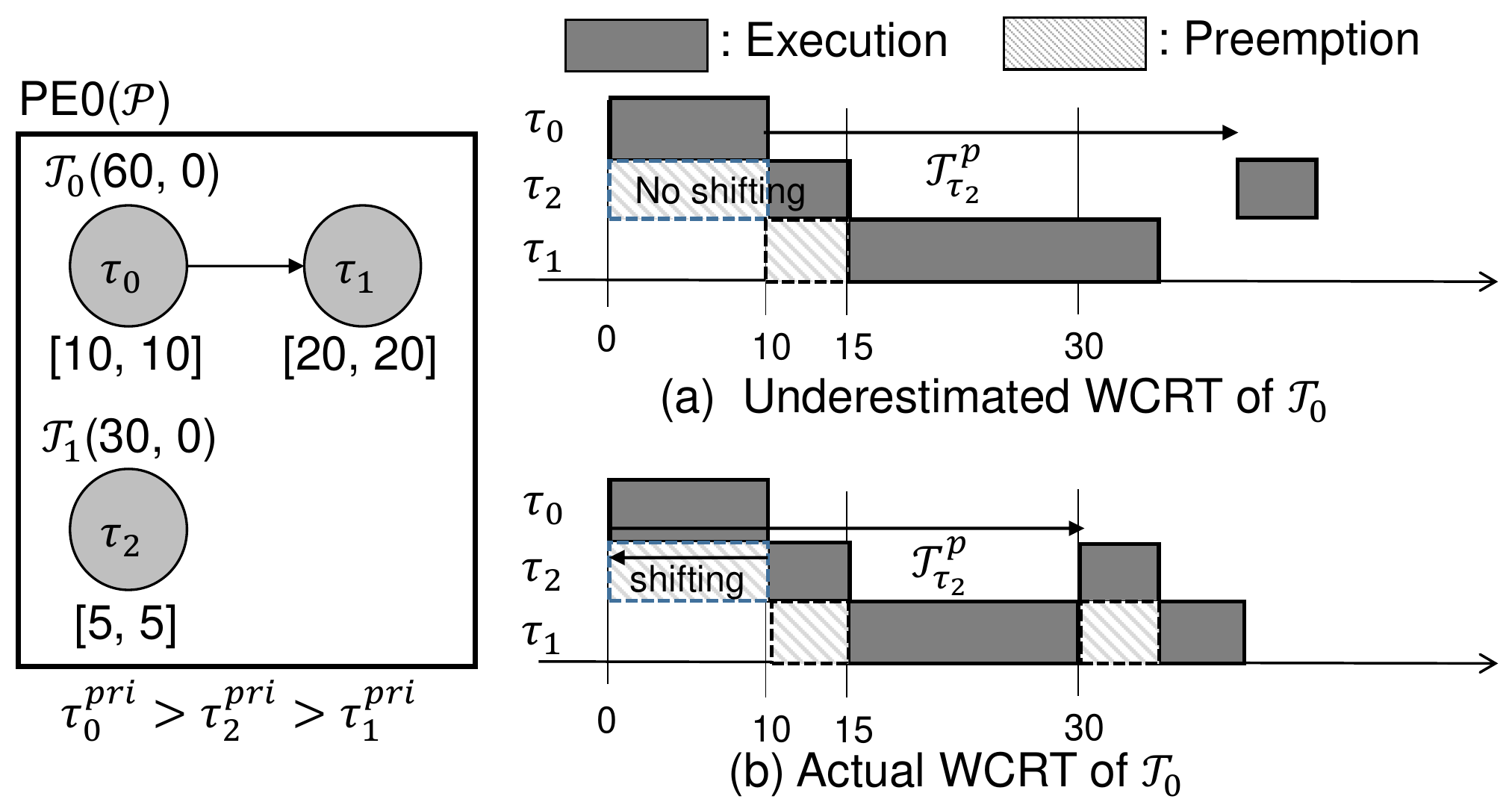}}
\caption{An underestimated WCRT example in the Y\&W method}
\label{fig:Ex_underestimation1}
\end{figure}

Another under-estimate example of Fig. \ref{fig:Ex_underestimation2} shows that period shifting and phase adjustment are inter-dependent while the Y\&W method treats them separately. In the Y\&W method, $\fphase{1}{3}$ and $\rphase{2}{3}$ are 60 assuming that task $\task{3}$ is first released at time 50 simultaneously with $\task{1}$ and the next release will be 100 time units after ignoring the period shifting effect. But the worst case of preemption occurs when $\task{3}$ is first released at time 0. Then the next release of $\task{3}$ can appear after $\task{2}$ executes 10 time units as shown in Fig. \ref{fig:Ex_underestimation2}. So, the WCRT of $\graph{0}$ is 110 since $\task{3}$ appears once per $\graph{0}$ execution in the Y\&W method while it is 130 since $\task{3}$ appears twice per $\graph{0}$ execution in the worst case. Motivated from this example, the proposed technique considers period shifting and phase adjustment holistically, which will be explained in the next section.

\begin{figure}[ht]
\centerline{\includegraphics[width=9.5cm]{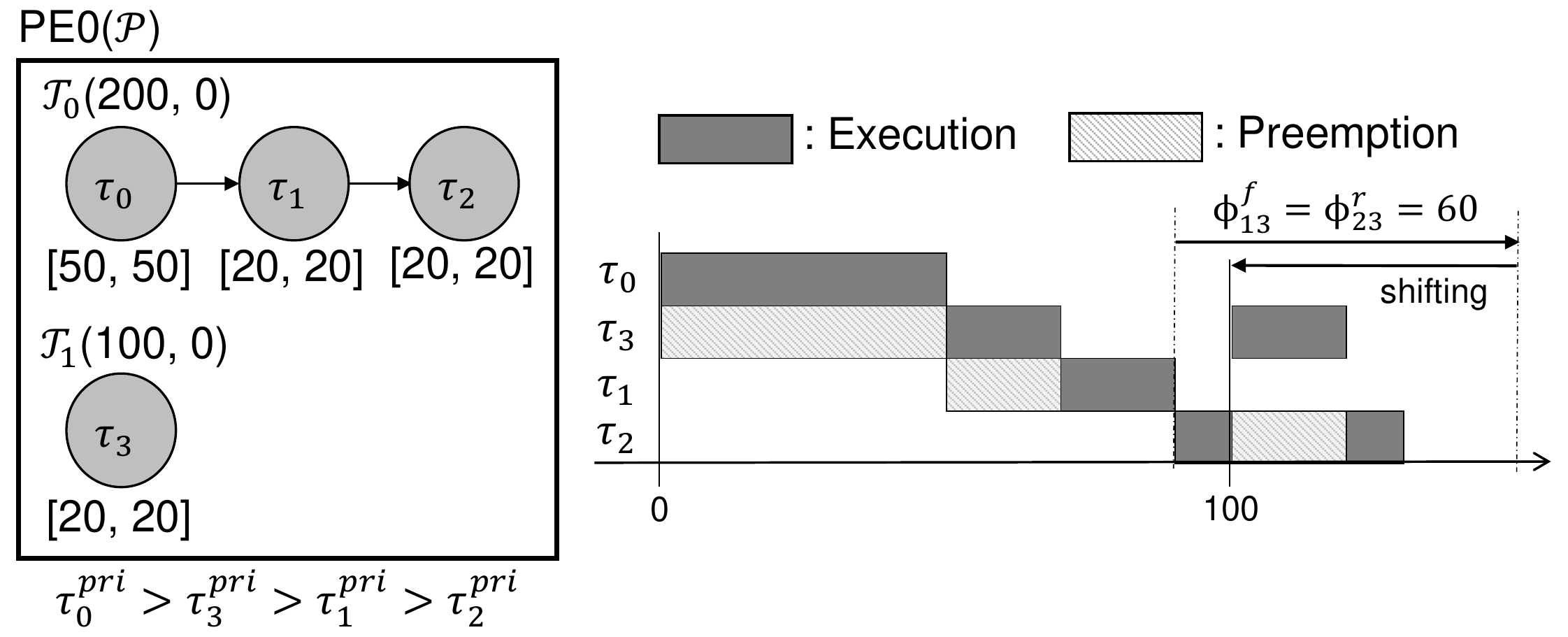}}
\caption{An example of incorrect phase adjustment}
\label{fig:Ex_underestimation2}
\end{figure}

\subsection{Overestimation made by the Y\&W method}

\begin{figure}[ht]
\centerline{\includegraphics[width=9.5cm]{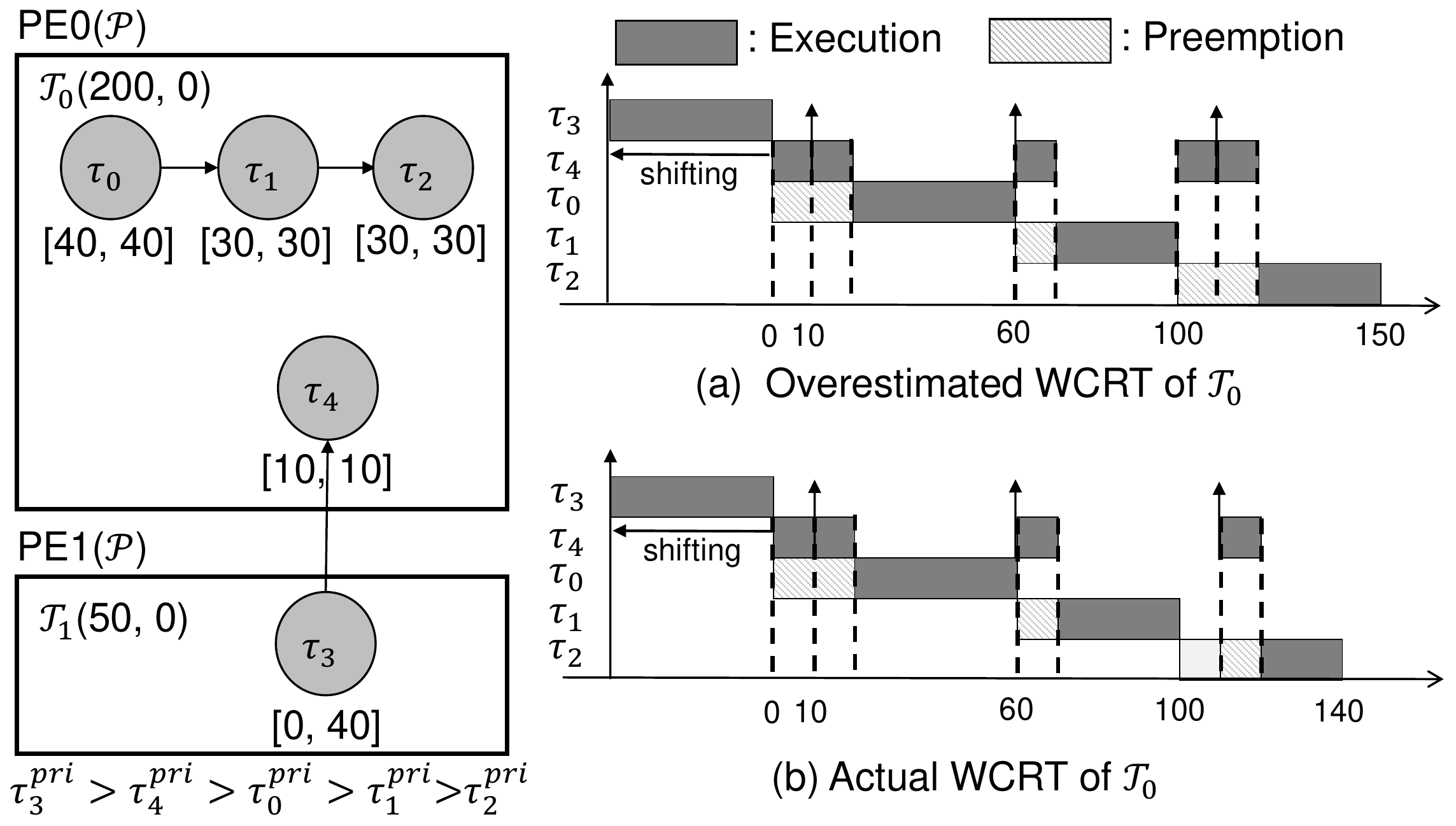}}
\caption{Overestimation due to repeated period shifting}
\label{fig:Ex_overestimation1}
\end{figure}

In the period shifting technique, the release time difference is added to the response time in the preemption delay computation. It should be applied only once in a sequence of releases of the preempting task. But the Y\&W method applies period shifting to all tasks independently, as shown in the example of Fig. \ref{fig:Ex_overestimation1} where period shifting is applied to all tasks, $\task{0}$ to $\task{2}$, in $\graph{0}$. As a result $\task{4}$ makes 5 preemptions in total to make the WCRT of the task graph be 150. But the actual WCRT is 140 as shown in Fig. \ref{fig:Ex_overestimation1} (b) since $\task{4}$ can make 4 preemptions at most. In the Y\&W method, period shifting value for $\task{4}$ is 40 because of the execution time variation of its predecessor $\task{3}$. Then $\task{4}$ preempts $\task{0}$ twice to make the WCRT of $\task{0}$ be $60=40+\lceil{\frac{(60-0+40)}{50}}\rceil\cdot10$. $\rphase{1}{4}$ is computed to $40=(0-60)\bmod50$, and the WCRT of $\task{1}$ is $40=30+\lceil{\frac{(40-40+40)}{50}}\rceil\cdot10$. Since $\rphase{2}{4}$ becomes $0=40-40$, the WCRT of $\task{2}$ is $50=30+\lceil{\frac{(50-0+40)}{50}}\rceil\cdot10$, experiencing two preemptions. By integrating period shifting into phase adjustment we could improve the tightness of the WCRT by removing the problem of redundant application of period shifting, which will be explained in the next section.

\begin{figure}[ht]
\centerline{\includegraphics[width=9cm]{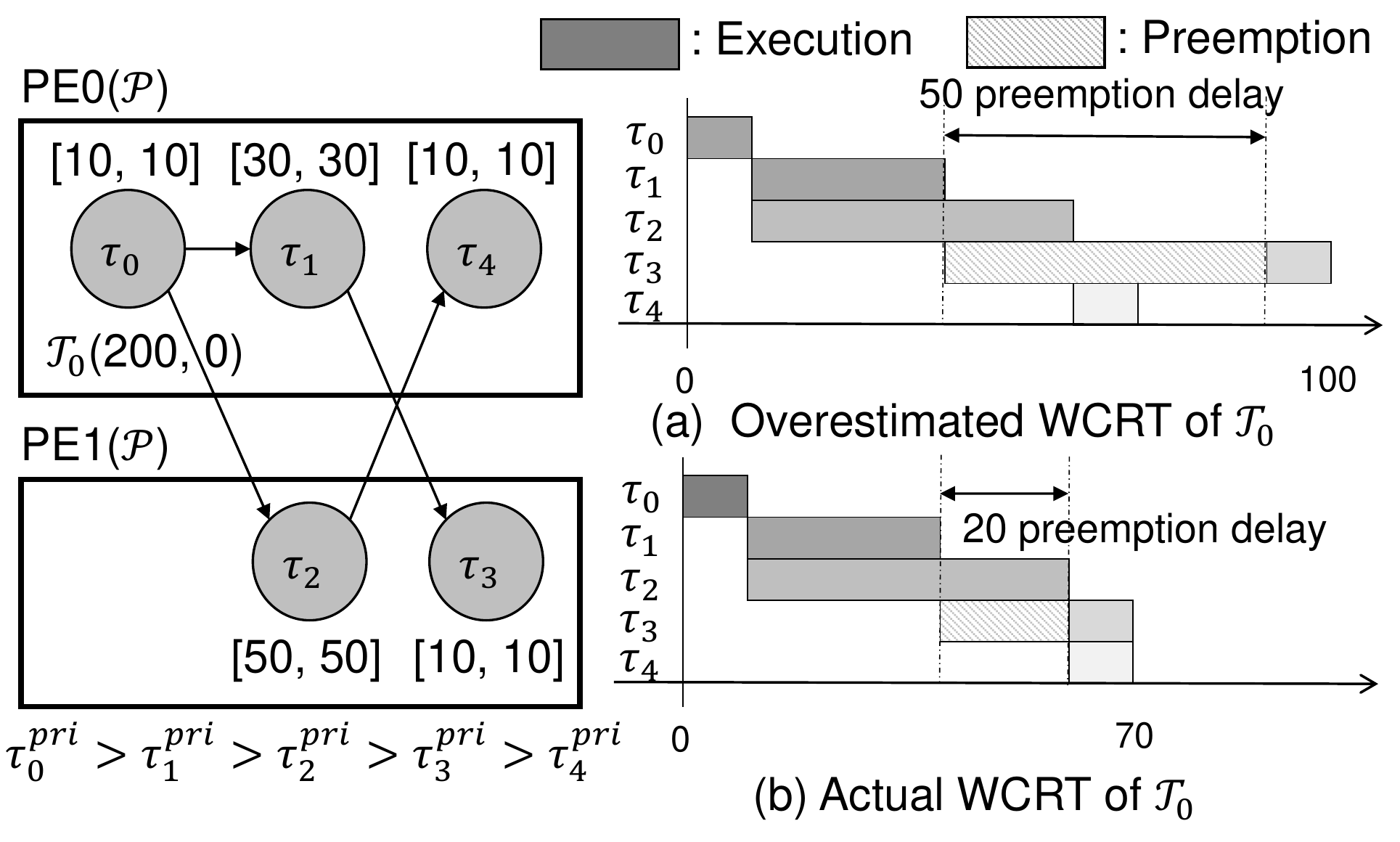}}
\caption{An example of overestimation (Y\&W method)}
\label{fig:Ex_overestimation2}
\end{figure}

Another cause of overestimation comes from the fact that the Y\&W method does not consider partial preemptions. In the example of Fig. \ref{fig:Ex_overestimation2}, $\task{1}$ cannot preempt $\task{4}$ since two tasks are separated. On the other hand, $\task{2}$ and $\task{3}$ are not separated because $\task{3}$ can be released during the execution of $\task{2}$. In the Y\&W method, however, the WCET of $\task{3}$ is wholly added to the preemption delay of $\task{2}$. Partial preemption may occur between tasks in the same task graph as shown in this example, which is not considered in the baseline response time analysis. In the proposed analysis, however, we perform scheduling of tasks in the same task graph so that we could detect this kind of partial preemption precisely.

\section{Proposed Analysis Technique: Hybrid Performance Analysis} \label{sec:HPA}

The proposed technique called HPA (hybrid performance analysis) extends and combines the STBA approach and the RTA method: the former is to account for interference between tasks in a same task graph and the latter for interference from the other task graphs.

Fig. \ref{fig:HPA_overallflow} shows the algorithm flow of the proposed technique. First, we compute three pairs of time bound information for each task: release time bound($\minR{i}$, $\maxR{i}$), start time bound ($\minS{i}$, $\maxS{i}$), and finish time bound ($\minF{i}$, $\maxF{i}$). Unlike the STBA technique [Kim et al. 2013]\ignore{\cite{ref5}} that assumes to schedule all task graphs together, the proposed technique schedules each task graph separately at the time bound computation step. The interference from the other task graphs is considered by the holistic phase adjustment technique based on the period shifting amount computed in the previous iteration. A period shifting amount is updated after time bound computation. This process is repeated until every value is converged.

\begin{figure}[ht]
\centerline{\includegraphics[width=11.5cm]{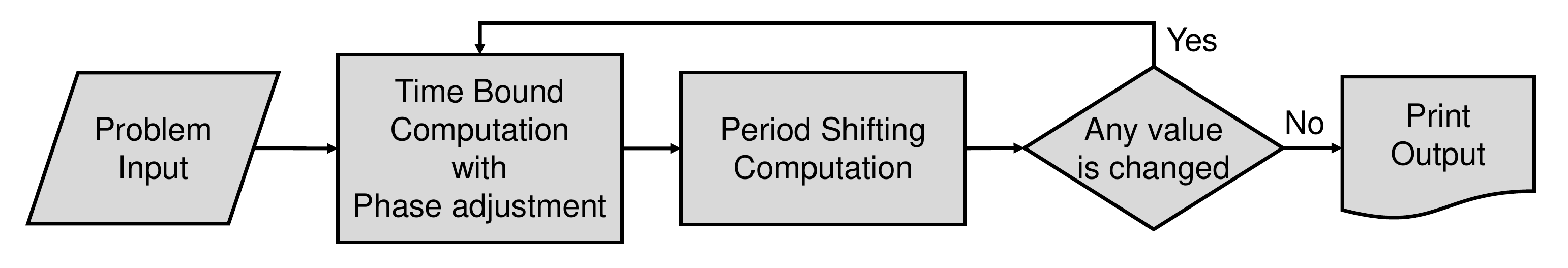}}
\caption{The proposed HPA technique overall flow}
\label{fig:HPA_overallflow}
\end{figure}

In this section, we explain the key techniques of the proposed analysis and discuss how to achieve a safe and tighter bound of the WCRT. At first, how to compute time bounds is explained. Secondly, we derive and optimize a conservative bound of a period shifting amount from $\task{i}$ to $\task{t}$ which bounds the maximum preemption count of $\task{i}$ between release time and the maximum start time of $\task{t}$. Finally, \emph{period shifting} is combined with \emph{phase adjustment} technique in the proposed \emph{holistic phase adjustment} technique.

\subsection{Time Bound Computation}

Let $\release{t}$, $\start{t}$, and $\finish{t}$ denote the actual release time, start time, and
finish time of task $\task{t}$. In our task graph model, the release time of a task is the maximum finish time of its immediate predecessors , as summarized in the following definition.

\begin{defn}
\[
\release{t} =\left\{\begin{array}{rl}
	\graphRelease{\task{t}}, & \textrm{if $\task{t}$ is a source task}\\
	\max_{\task{p}\in{pred(\task{t})}}{\finish{p}}, & \textrm{otherwise}\\
\end{array} \right.
\]
\end{defn}

where $\graphRelease{\task{t}}$ denotes the release time of the task graph. Then the minimum ($\minR{t}$) and maximum ($\maxR{t}$) release time bound pair is computed as follows:

\begin{equation} \label{eq:Eq10}
\minR{t}=\left\{\begin{array}{rl}
	0, & \textrm{if $\task{t}$ is a source task}\\
	\max_{\task{p}\in{pred(\task{t})}}{\minF{p}}, & \textrm{otherwise}\\
\end{array} \right.
\end{equation}

\begin{equation} \label{eq:Eq11}
\maxR{t}=\left\{\begin{array}{rl}
	\jitter{\task{t}}, & \textrm{if $\task{t}$ is a source task}\\
	\max_{\task{p}\in{pred(\task{t})}}{\maxF{p}}, & \textrm{otherwise}\\
\end{array} \right.
\end{equation}

The earliest and the latest release times of a non-source task are defined as the maximum value among the earliest and the latest finish times of predecessors, respectively since it becomes executable only after all predecessor tasks are finished.

\begin{lemm} \label{theorem3}
$\minR{t}$ and $\maxR{t}$ are conservative, or \emph{ $\minR{t} \le \release{t} \le \maxR{t}$}.
\end{lemm}

\begin{proof}
If $\task{t}$ is a source task, $\minR{t}=0\le\release{t}\le\jitter{\task{t}}=\maxR{t}$ since $0 \le \graphRelease{\task{t}} \le \jitter{\task{t}}$.

For non-source task $\task{t}$, assume that it holds for all predecessor tasks of $\task{t}$. Since $ \minR{t} = \max_{\task{p}\in{pred(\task{t})}}{\minF{p}} \le \max_{\task{p}\in{pred(\task{t})}}{\finish{p}} = \release{t} \le \max_{\task{p}\in{pred(\task{t})}}{\maxF{p}} = \maxR{t}$, $\minR{t} \le \release{t} \le \maxR{t}$.

By induction, the lemma holds. Q.E.D.
\end{proof}

For task $\task{t}$ to start, it should be already released and the processor must be available: The start time of $\task{t}$ is not smaller than the release time and the maximum time among finish times of tasks that have higher priority, start earlier, and finish after task $\task{t}$ is released. Formally, the start time is defined as follows:

\begin{defn}
$\start{t}=\max(\release{t}, \max_{\task{s}\in{\Gamma_{\task{t}}}}{\finish{s}})$
\end{defn}

where set ${\Gamma_{\task{t}}}$ is defined as ${\Gamma_{\task{t}}}=\{\task{s}|\proc{s}=\proc{t},\pri{s}>\pri{t},\release{t}<\finish{s},\start{s} \le \start{t}\}$ for the preemptive scheduling policy, and ${\Gamma_{\task{t}}}=\{\task{s}|\proc{s}=\proc{t},(\pri{s}>\pri{t},\release{t}<\finish{s},\start{s}\le\start{t}) or (\pri{s}<\pri{t}, \start{s}<\release{t}<\finish{s})\}$ for the non-preemptive scheduling policy.

One the other hand, the earliest start time $\minS{t}$ is formulated as follows:

\begin{equation} \label{eq:Eq12}
\minS{t}=\max(\minR{t}, \max_{\task{s}\in\condset{A}{t}}{\minF{s}})
\end{equation}

where set $\condset{A}{t}$ for the preemptive scheduling policy is defined as $\condset{A}{t}=\{\task{s}|\task{s}\in\graph{\task{t}},\proc{s}=\proc{t},\pri{s}>\pri{t},\minR{t}<\minF{s},\maxS{s}\le\minS{t}\}$ and for the non-preemptive scheduling policy as $\condset{A}{t}=\{\task{s}|\task{s}\in\graph{\task{t}},\proc{s}=\proc{t},(\pri{s}>\pri{t},\minR{t}<\minF{s},\maxS{s}\le\minS{t}) or (\pri{s}<\pri{t},\maxS{s}<\minR{t}<\minF{s})\}$. If higher priority task $\task{s}$ always starts before $\minS{t}$ and the earliest finish time of $\task{s}$ is later than $\minR{t}$, $\task{t}$ should wait for the completion of $\task{s}$. In case a non-preemptive scheduling is used, a lower priority task that starts before $\minR{t}$ is included. Note that we only consider the tasks in the same task graph since the other tasks can appear at any time, so they should be ignored for conservative estimation of the minimum start time.

To estimate the maximum start time $\maxS{t}$ for conservative estimation, we should consider all possible preemptions.

\begin{equation} \label{eq:Eq13}
\maxS{t}=\maxR{t}+Delay_t^l+Delay_t^h
\end{equation}

where $Delay_t^l$ and $Delay_t^h$ denote the amounts of preemption between the release time and the start time by lower and higher priority tasks respectively. For the premeptive scheduling policy, $Delay_t^l$ is zero. In case a lower priority task is running when $\task{t}$ is released, $\task{t}$ should wait until the current lower priority task finishes in the non-preemptive scheduling policy, which is accounted as follows:

\begin{equation} \label{eq:Eq14}
Delay_t^l=\left\{\begin{array}{rl}
	0, & \textrm{if $\forall_{\task{p}\in{pred(\task{t})}}{(\proc{p}=\proc{t})}$}\\
	\max\left(\begin{array}{c}
		\max\limits_{\task{s}\in{\condset{B}{t}}}{\min(\WCET{s},\maxF{s}-\maxR{t})},\\
		\max\limits_{\task{s}\in{\condset{C}{t}}}{\WCET{s}}
	\end{array}\right), & \textrm{otherwise}\\
\end{array} \right.
\end{equation}

where $\condset{B}{t}=\{\task{s}|\task{s}\in\graph{\task{t}},\proc{s}=\proc{t},\pri{s}<\pri{t},\minS{s}<\maxR{t}<\maxF{s}\}$ and $\condset{C}{t}=\{\task{s}|\task{s}\not\in\graph{\task{t}},\proc{s}=\proc{t},\pri{s}<\pri{t}\}$. Set $\condset{B}{t}$ includes lower priority tasks in the same task graph that may start earlier than $\task{t}$ and delay the start time of $\task{t}$. Note that partial blocking is considered in the formula by $\maxF{s}-\maxR{t}$, which can be smaller than $\WCET{s}$. On the other hand, set $\condset{C}{t}$ includes all lower priority tasks in the other task graphs. Since they can appear at any time, we take the maximum WCET for conservative estimation. In case every predecessor is mapped to the same PE, $Delay_t^l$ is zero.

$Delay_t^h$ is commonly formulated for both scheduling policies as follows:

\begin{eqnarray} \label{eq:Eq15}
Delay_t^h=\sum_{\task{s}\in\condset{D}{t}}{\min(\WCET{s}, \maxF{s}-\maxR{t})}+{}\qquad\qquad\qquad\nonumber\\
\qquad\qquad\qquad\sum_{\task{s}\in\condset{E}{t}}{\bigg\lceil{\frac{\max(0,\maxS{t}-\maxR{t}+1-\rphase{t}{s})}{\period{\task{s}}}}\bigg\rceil\cdot\WCET{s}}
\end{eqnarray}

where $\condset{D}{t}=\{\task{s}|\task{s}\in\graph{\task{t}},\proc{s}=\proc{t},\pri{s}>\pri{t},\minS{s}\le\maxS{t},\maxR{t}<\maxF{s}\}$ and $\condset{E}{t}=\{\task{s}|\task{s}\not\in\graph{\task{t}},\proc{s}=\proc{t},\pri{s}>\pri{t}\}$. Set $\condset{D}{t}$ includes higher priority tasks in the same task graph that can possibly delay the start time of $\task{t}$. Partial preemption is considered similarly to $Delay_t^l$ formulation. For the example of Fig. \ref{fig:Ex_overestimation2}, we precisely compute the preemption delay from $\task{2}$ to $\task{3}$ as $\maxF{2}-\maxR{3}=60-40=20$, which is less that $\WCET{2}$. For the higher priority tasks in the other task graphs, we use a similar formula as the response time analysis to compute the maximum preemption delay. The notation $\rphase{t}{s}$ corresponds to the request phase adjustment that computes the minimum distance from the release time of $\task{t}$ to the next release time of a preempting task $\task{s}$. How to formulate the request phase $\rphase{t}{s}$ will be explained in the next subsection.

\begin{lemm} \label{theorem4}
$\minS{t}$ and $\maxS{t}$ are conservative, or \emph{$\minS{t} \le \start{t} \le \maxS{t}$}.
\end{lemm}

\begin{proof}
First we prove that $\minS{t}\le\start{t}$.

Since $\minR{t} \le \release{t}$, we need to prove $\max_{\task{s}\in\condset{A}{t}}{\minF{s}} \le \start{t}$.
$$\max_{\task{s}\in\condset{A}{t}}{\minF{s}} \le \max(\max_{\task{s}\in{\Gamma_{\task{t}}}}{\minF{s}}, \max_{\task{s}\in\condset{A}{t}- {\Gamma_{\task{t}}}}{\minF{s}}).$$
Since $\condset{A}{t}- {\Gamma_{\task{t}}} \subseteq \{\task{s}|\task{s}\in\graph{\task{t}},\proc{s}=\proc{t},\release{t}\ge\finish{s}\}$ for both preemptive and nonpreemptive cases, and $\finish{s} \ge \minF{s}$, $\max_{\task{s}\in\condset{A}{t}- {\Gamma_{\task{t}}}}{\minF{s}} \le \release{t}$.
$$\max_{\task{s}\in\condset{A}{t}}{\minF{s}} \le \max(\max_{\task{s}\in{\Gamma_{\task{t}}}}{\finish{s}}, \release{t}) = \start{t}.$$

Next, We prove that $\start{t}\le\maxS{t}=\maxR{t}+Delay_t^l+Delay_t^h$ by showing that no task will delay the start of $\task{t}$ without being considered in $Delay_t^l$ or $Delay_t^h$ computation.

\begin{enumerate}
\item Suppose actual preemption amount by lower priority tasks is larger than $Delay_t^l$. For preemptive scheduling policy, it is impossible because a lower priority task cannot preempt $\task{t}$. For nonpreemptive scheduling policy, only one lower priority task can delay the execution of $\task{t}$. For $\task{s}\not\in\graph{\task{t}}$, since $\condset{C}{t}$ includes all lower priority tasks, it is not possible for the preemption amount to be larger than the maximum worst case execution time of tasks in $\condset{C}{t}$.

Consider $\task{s}\in \graph{\task{t}}$. $\task{s}$ can delay $\task{t}$ by $\finish{s}-\release{t}$ at most if $\task{s}$ starts before $\task{t}$ is released or $\start{s} < \release{t}$.
$$\release{t} + \max_{\task{s}\in\graph{\task{t}}, \finish{s}-\release{t}> 0}{(\finish{s}-\release{t})}$$
$$\le \maxR{t}+ \max_{\task{s}\in\graph{\task{t}}, \finish{s}-\maxR{t} > 0}{(\finish{s}-\maxR{t})}$$
$$\le \maxR{t} + \max_{\task{s}\in\graph{\task{t}}, \maxR{t} < \maxF{s}}{(\maxF{s}-\maxR{t})}.$$
Since $\{\task{s}| \start{s} < \release{t}\} \subseteq \{\task{s} | \minS{s} < \maxR{t}\}$ and $\{\task{s}|\task{s}\in\graph{\task{t}},\proc{s}=\proc{t},\pri{s}<\pri{t},\maxR{t}<\maxF{s}\} \cap \{\task{s}|\task{s}\in\graph{\task{t}},\proc{s}=\proc{t},\pri{s}<\pri{t},\minS{s}<\maxR{t}\}\subseteq \condset{B}{t}$,
$$\release{t} + \max\limits_{\task{s}\in\graph{\task{t}}, \finish{s}-\release{t}> 0}{(\finish{s}-\release{t})} \le \maxR{t} + \max\limits_{\task{s}\in\condset{B}{t}}{(\maxF{s}-\maxR{t})}.$$
Since $\task{s}$ cannot be executed longer than $\WCET{s}$,
$$\release{t} + \max_{\task{s}\in{\condset{B}{t}}}{\min(\WCET{s},\finish{s}-\release{t})} \le \maxR{t} + Delay_t^l.$$
Finally, consider the case when all predecessors are mapped to the same processor. Since there is no time interval between the finish time of the latest predecessor task $\task{p}$ and the release time of $\task{t}$, no lower prioirty task can start after $\task{p}$ finishes and before $\task{t}$ is released. \label{theorem3_2_1}

\item Consider higher priority tasks that may preempt $\task{t}$. For $\task{s}\in\graph{\task{t}}$ and $\task{s}\in\condset{D}{t}$, the preemption amount is bounded by $\min(\WCET{s}, \maxF{s}-\maxR{t})$ similarly to the proof of the case $\task{s}\in\graph{\task{t}}$ and $\task{s}\in\condset{B}{t}$ in \ref{theorem3_2_1}.

Suppose there is a task $\task{s}$ which can preempt $\task{t}$ such that $\task{s}\in\graph{\task{t}}$ and $\task{s}\not\in\condset{D}{t}$ . Then $\maxS{t}<\minS{s}$ or $\maxF{s}\le\maxR{t}$. $\task{s}$ always starts later than $\task{t}$ since $\maxS{t}\le\minS{s}$, or $\task{s}$ finishes earlier than $\task{t}$ since $\maxF{s}\le\maxR{t}$. No preemption may be occurred so that there exists no such task.

For $\task{s}\not\in\graph{\task{t}}$, the proof is trivial if the request phase is conservatively given. The conservativeness of the request phase, $\rphase{t}{s}$ will be proven in Lemma \ref{theorem2}. \label{theorem3_2_2}
\end{enumerate}

By \ref{theorem3_2_1} and \ref{theorem3_2_2}, $\start{t}\le\maxS{t}$. Q.E.D.
\end{proof}

The minimum finish time $\minF{t}$ is formulated as follows:

\begin{equation} \label{eq:Eq16}
\minF{t}=\minS{t}+\BCET{t}+Preempt_t^B
\end{equation}

where $Preempt_t^B$ represents the unavoidable (best-case) preemption delay that is zero for the non-preemptive scheduling policy. For the preemptive scheduling policy, $Preempt_t^B$ becomes

\begin{equation} \label{eq:Eq17}
Preempt_t^B=\sum_{\task{s}\in\condset{F}{t}}{\BCET{s}}
\end{equation}

where $\condset{F}{t}=\{\task{s}|\task{s}\in\graph{\task{t}},\proc{s}=\proc{t},\pri{s}>\pri{t},\minS{t} \le \minS{s} \le \maxS{s} \le \minF{t}\}$. $\condset{F}{t}$ is a set of higher priority tasks which always start to execute and preempt $\task{t}$ during $\task{t}$ is running from $\minS{t}$ to $\minF{t}$.

The maximum finish time $\maxF{t}$ is formulated as follows:

\begin{equation} \label{eq:Eq18}
\maxF{t}=\maxS{t}+\WCET{t}+Preempt_t^W
\end{equation}

where $Preempt_t^W$ represents the worst-case preemption delay that is zero for the non-preemptive scheduling policy. $Preempt_t^W$ for the preemptive scheduling policy is formulated as follows:

\begin{equation} \label{eq:Eq19}
Preempt_t^W=\sum_{\task{s}\in\condset{G}{t}}{\WCET{s}}+\sum_{\task{s}\in\condset{E}{t}}{\bigg\lceil{\frac{\max(0,\maxF{t}-\maxS{t}-\sphase{t}{s})}{\period{\task{s}}}}\bigg\rceil\cdot\WCET{s}}
\end{equation}

where $\condset{G}{t}=\{\task{s}|\task{s}\in\graph{\task{t}},\proc{s}=\proc{t},\pri{s}>\pri{t},\maxS{t}<\minS{s}\le\maxF{t}\}$, indicating a set of higher priority tasks which can appear during the execution of $\task{t}$. The notation $\sphase{t}{s}$ is the start phase that is the the minimum distance from the start time of $\task{t}$ to the next release time of a preempting task $\task{s}$. For tasks in the other task graphs, the maximum possible preemption delay is computed similarly to (\ref{eq:Eq15}). The start phase will be explained in the next subsection.

\begin{lemm} \label{theorem5}
$\minF{t}$ and $\maxF{t}$ are conservative, or \emph{$\minF{t} \le \finish{t} \le \maxF{t}$}.
\end{lemm}

\begin{proof}
(Nonpreemptive) If non-preemptive scheduling is used then no task may preempt $\task{t}$. Hence the finish time is the sum of the start time and the execution time.

(Preemptive) First, we prove that $\minF{t}\le\finish{t}$. Let $\condset{W}{t}$ be $\{\task{s}|\task{s}\in\graph{\task{t}}, \proc{s}=\proc{t}, \pri{s} > \pri{t}\}$.

$\finish{t} \ge \start{t} + \BCET{t} + \sum_{\task{s}\in \condset{W}{t}, \start{t} \le \start{s} \le \finish{t}}{\BCET{s}}$ where the last term accounts for tasks that start to execute during $\task{t}$ is running, and preempt $\task{t}$.

Since $\sum_{\task{s}\in \condset{W}{t}, \start{t} \le \start{s} \le \finish{t}}{\BCET{s}} \le \finish{t}-\start{t}$,\\
and $\sum_{\task{s}\in \condset{W}{t}, \minS{t} \le \start{s} \le \finish{t}}{\BCET{s}} \le \sum_{\task{s}\in \condset{W}{t}, \start{t} \le \start{s} \le \finish{t}}{\BCET{s}}+(\start{t}-\minS{t})$,
$$\finish{t} \ge \start{t} + \BCET{t} + \sum_{\task{s}\in \condset{W}{t}, \start{t} \le \start{s} \le \finish{t}}{\BCET{s}}$$
$$\ge \minS{t} + \BCET{t} + \sum_{\task{s}\in \condset{W}{t}, \minS{t} \le \start{s} \le \finish{t}}{\BCET{s}}$$
$$\ge \minS{t} + \BCET{t} + \sum_{\task{s}\in \condset{W}{t}, \minS{t} \le \minS{s} \le \maxS{s} \le \finish{t}}{\BCET{s}}$$
since $\minS{s} \le \start{s} \le \maxS{s}.$
$$\finish{t} \ge \minF{t} = \minS{t} + \BCET{t} + \sum_{\task{s}\in \condset{W}{t}, \minS{t} \le \minS{s} \le \maxS{s} \le \minF{t}}{\BCET{s}}.$$
Second, we prove that $\finish{t}\le\maxF{t}$ by contradiction.

For $\task{s}\in\graph{\task{t}}$, suppose there is some $\task{s}\not\in\condset{G}{t}$ that satisfies $\task{s}\in\graph{\task{t}}$ and can preempt $\task{t}$. Then $\minS{s}\le\maxS{t}$ or $\maxF{t}<\minS{s}$. The tasks that satisfies $(\task{s}\not\in\condset{G}{t},\task{s}\in\graph{\task{t}},\minS{s}\le\maxS{t},\maxR{t}<\maxF{s})$ belongs to the $\condset{D}{t}$, so that the amount of preemption from those tasks are already included in $\maxS{t}$ and so in $\maxF{t}$. The tasks that satisfies $(\task{s}\not\in\condset{G}{t},\task{s}\in\graph{\task{t}},\minS{s}\le\maxS{t},\maxF{s}\le\maxR{t})$ cannot preempt $\task{t}$ since $\task{s}$ finishes before $\maxR{t}$. The tasks that satisfies $(\task{s}\not\in\condset{G}{t},\task{s}\in\graph{\task{t}},\maxF{t}<\minS{s})$ cannot preempt $\task{t}$ since $\task{s}$ starts after $\maxF{t}$. For $\task{s}\not\in\graph{\task{t}}$, the proof is trivial since we assume that the start phase is conservatively given. The conservativeness of the start phase, $\sphase{t}{s}$ will be proven in Lemma \ref{theorem2}. Q.E.D.
\end{proof}

After determining all time bounds of tasks, we compute the WCRT of each task graph $\graph{}$ as follows:

\begin{equation} \label{eq:Eq20}
\mathcal{R}_{\graph{}}=\max_{\task{s}\in\graph{}}\maxF{s}
\end{equation}

\begin{theo} \label{theorem6}
The HPA technique guarantees the conservativeness of every schedule time bound when it is converged.
\end{theo}

\begin{proof}
Theorem \ref{theorem6} is proved by Lemma \ref{theorem3}, Lemma \ref{theorem4}, and Lemma \ref{theorem5}. Q.E.D.
\end{proof}

\subsection{Period Shifting Computation}

\begin{figure}[ht]
\centerline{\includegraphics[width=9cm]{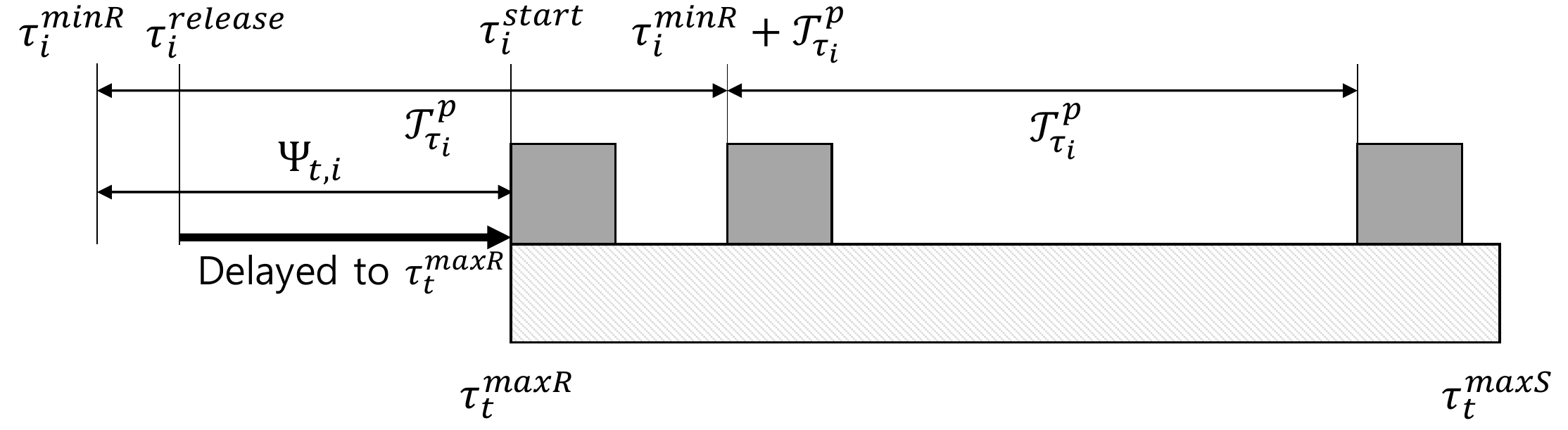}}
\caption{The worst case scenario of the preemption by $\task{i}$ to $\task{t}$}
\label{fig:App_periodshifting}
\end{figure}

When computing the maximum time bounds, we have to consider the worst case scenario of preemptions in the RTA analysis. Fig. \ref{fig:App_periodshifting} shows this scenario. Suppose that the target task $\task{t}$ is released at $\maxR{t}$. The worst case preemptions of a higher priority task $\task{i}$ occur when its start time is aligned with the release time of $\task{t}$ and the second request appears with the shortest interval from $\maxR{t}$, followed by later requests that appear periodically from the second request. If the period shifting amount is denoted by $\pshift{t}{i}$, the next start time of $\task{i}$ will be $\max (\maxR{t}+\WCET{i}, \maxR{t}-\pshift{t}{i}+\period{t})$.

\begin{figure}[ht]
\centerline{\includegraphics[width=11cm]{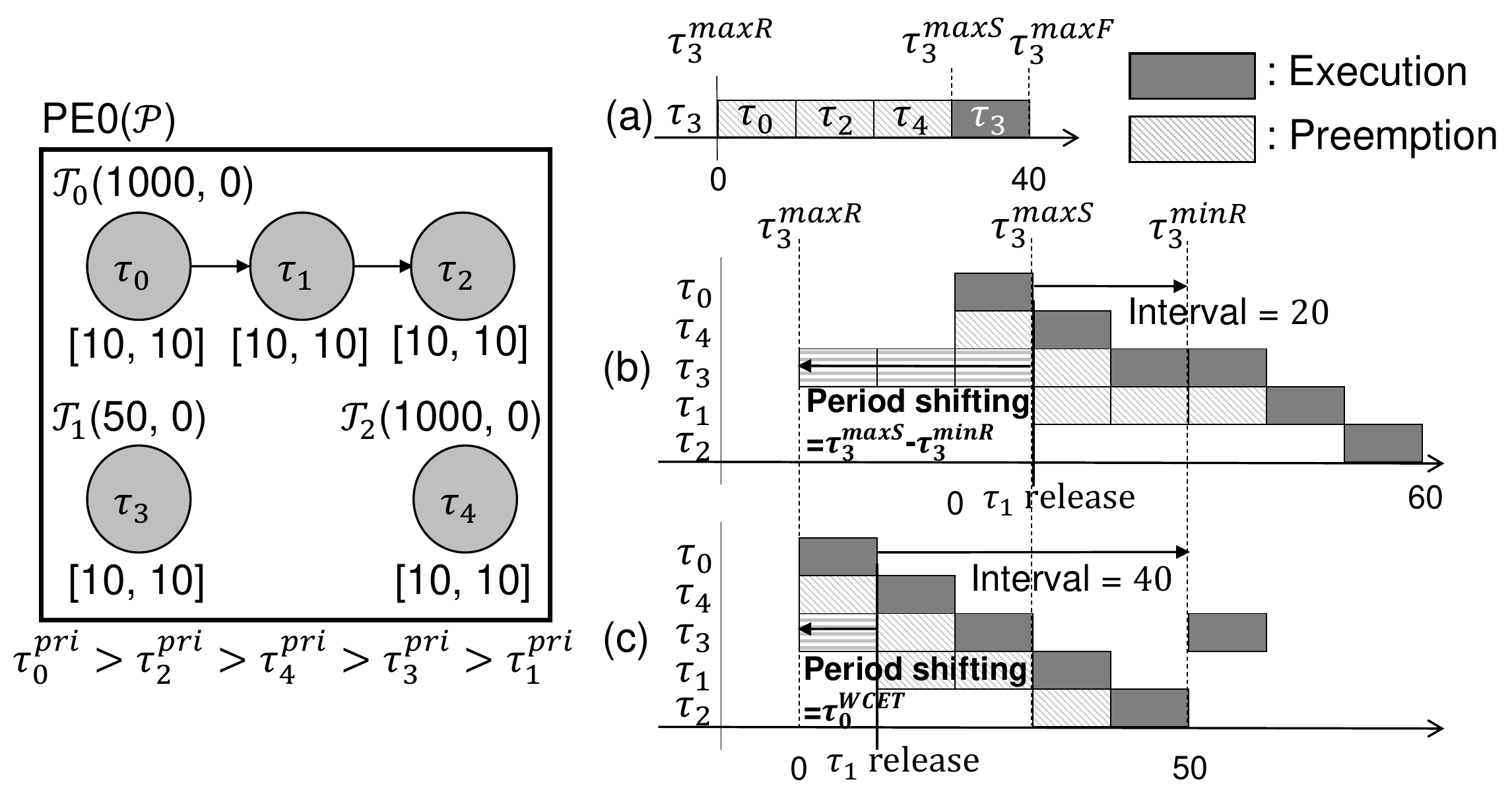}}
\caption{An example for period shifting computation. (a) Possible preemptions to $\task{3}$ between $\maxR{3}$ and $\maxS{3}$, (b) estimated WCRT of $\graph{0}$ with period shifting $\maxS{3}-\minR{3}$, and (c) actual WCRT of $\graph{0}$}
\label{fig:Ex_periodshifting1}
\end{figure}

Recall that the Y\&W method defines $\pshift{t}{i}$ as $\maxR{i}-\minR{i}$ and it fails to find the worst case behavior as shown in the example of Fig. \ref{fig:Ex_underestimation1}. A naive revision is to define $\pshift{t}{i}$ as $\maxS{i}-\minR{i}$, by aligning the maximum start time of $\task{i}$ with the maximum release time of $\task{t}$. The worst case bursty requests occurs when the preempting task is released at $\maxR{t}-\pshift{t}{i}$. In the example of Fig. \ref{fig:Ex_underestimation1}, we can obtain the actual WCRT if $\maxS{2}$ is aligned at $\maxR{1}$, which makes $\pshift{1}{2}=\maxS{2}-\minR{2}=10$. However, it generates a loose bound for the example of Fig. \ref{fig:Ex_periodshifting1}. Fig. \ref{fig:Ex_periodshifting1} (a) shows the worst case preemption of $\task{3}$, considering all possible preemptions from the other tasks. From Fig. \ref{fig:Ex_periodshifting1} (a), it is known that $\maxS{3}-\minR{3}$ is 30. If this value is used as period shifting, a loose WCRT is obtained as shown in Fig. \ref{fig:Ex_periodshifting1} (b). The actual WCRT is computed when $\pshift{1}{3}$ is equal to $\WCET{0}$, as shown in the Fig. \ref{fig:Ex_periodshifting1} (c), which confirms that more elaborate preemption analysis is needed between $\maxR{3}$ and $\maxS{3}$. For the conservative but tight bound computation of period shifting amount, we classify possible preemptions between $\maxR{3}$ and $\maxS{3}$ into three categories. The first is the preemptions that affect the start times of both $\task{3}$ and $\task{1}$, which is the amount of $\WCET{4}$ in Fig. \ref{fig:Ex_periodshifting1}. These preemptions can be ignored for $\pshift{1}{3}$ since those are considered in the computation of $\maxS{1}$. The second is the preemptions from the tasks that finish earlier than $\task{1}$ such as $\task{0}$ in Fig. \ref{fig:Ex_periodshifting1}, which should be considered for the $\pshift{1}{3}$. The last is the preemptions from the task that starts later than $\task{1}$ such as $\task{2}$ that is a descendant of $\task{1}$. These also can be ignored since they cannot appear before $\maxS{1}$.

Based on these observations, we conservativley consider the tasks in the second category in the period shifting computation to get a tighter bound than $\maxS{i}-\minR{i}$: that is

\begin{eqnarray} \label{eq:Eq6}
\pshift{t}{i} = \maxR{i} - \minR{i} +\delta_{t,i}
\end{eqnarray}

where $\delta_{t,i} = \sum_{\pri{s} > \pri{i}, \task{s} \in \condset{E}{t}}{\Big(\Big\lceil \frac{R_{\task{t}} + \delta_{t,i} + \pshift{t}{s}}{\period{\task{s}}} \Big\rceil - \Big\lceil \frac{R_{\task{t}} - \rphase{t}{s}}{\period{\task{s}}} \Big\rceil \Big) \cdot \WCET{s}} + \sum_{\task{s} \in \graph{\task{t}}, \pri{s} > \pri{i}, (\maxR{t} - \delta_{t,i} \le \maxR{s} < \maxR{t} ~or~ \maxR{t} - \delta_{t,i} \le \maxR{s}-\period{\task{s}} < \maxR{t})}{\WCET{s}}$, and $R_{\task{t}} = \maxS{t} - \maxR{t}$.

\begin{lemm} \label{theorem1}
Period shifting is conservative, which means that $\Big\lceil{\frac{\maxF{t}-\maxR{t}+\pshift{t}{i}}{\period{\task{i}}}}\Big\rceil$ is an upper bound of preemption counts of $\task{i}$ in time window $[\maxR{t}, \maxF{t}]$.
\end{lemm}

\begin{proof}
Conservativeness of period shifting is achieved if $\delta_{t,i}$ of equation (\ref{eq:Eq6}) includes all tasks in the second category of tasks preempting $\task{i}$ before $\maxR{t}$.

At first, consider the contribution from $\task{s}\not\in\graph{\task{t}}$ to $\delta_{t,i}$. Since the maximum preemption count of $\task{t}$ by $\task{s} \in \condset{E}{t}$ is $\lceil \frac{t+\pshift{t}{s}}{\period{s}} \rceil$ for $t$ time duration, $\task{s}$ appears $\lceil \frac{R_{\task{t}} + \delta_{t,i} +\pshift{t}{s}}{\period{s}} \rceil$ times during $R_{\task{t}}+\delta_{t,i}$. On the other hand, $\lceil \frac{R_{\task{t}}-\rphase{t}{s}}{\period{s}} \rceil$ is the number of preemptions during $R_{\task{t}}$, which is considered in $\maxS{t}$ computation. Therefore the appearance of $\task{s}$ is bounded by $\Big\lceil \frac{R_{\task{t}} + \delta_{t,i} + \pshift{t}{s}}{\period{\task{s}}} \Big\rceil - \Big\lceil \frac{R_{\task{t}} - \rphase{t}{s}}{\period{\task{s}}} \Big\rceil$ times during $\delta_{t,i}$.

Second, Consider $\task{s} \in \graph{\task{t}}$. If $\release{s} > \maxR{t}$ then $\task{s}$ either preempts $\task{t}$ as well as $\task{i}$ (the first category) or is scheduled later than $\task{t}$ (the third category). So we need to consider only $\task{s}$ that is released between $\maxR{t}-\delta_{t,i}$ and $\maxR{t}$.
$\maxR{t} - \delta_{t,i} < \release{s} \le \maxR{s} < \maxR{t}$. So the worst case is to consider $\task{s}$ such that $\maxR{t} - \delta_{t,i} < \maxR{s} < \maxR{t}$. If $\maxR{t} - \delta_{t,i}$ becomes negative, , the previous instance of $\task{s}$ should be considered with the following condition: $\maxR{t}-\delta_{t,i} < \maxR{s}-\period{\task{s}} < \maxR{t}$. Q.E.D.
\end{proof}

Consider Fig. \ref{fig:Ex_underestimation1}. $\delta_{1,2} = \WCET{0} = 10$ since $\task{0}$ satisfies the condition of the second term of $\delta_{1,2}$, $\maxR{1} - \delta_{1,2} = 10 - 10 = 0 = \maxR{0}$. Hence  $\pshift{1}{2} = \maxR{2}- \minR{2} + \delta_{1,2} = 0 - 0 + 10 = 10$.

Consider Fig. \ref{fig:Ex_periodshifting1}. $\delta{1,3} = \WCET{0} = 10$. Since $\maxR{1} - \delta_{1,3} = 10 - 10 = 0 = \maxR{0}$, $\task{0}$ is included in the second term of $\delta_{1,3}$. On the other hand, $\task{2}$ is excluded in $\delta_{1,3}$ since $10 = \maxR{1} < \maxR{2} = 40$. $\task{4}\not\in\graph{\task{1}}$ has no contribution to $\delta_{t,i}$ since $\Big(\Big\lceil \frac{R_{\task{1}} + \delta_{1,3} + \pshift{1}{4}}{\period{\task{4}}} \Big\rceil - \Big\lceil \frac{R_{\task{1}} - \rphase{1}{4}}{\period{\task{4}}} \Big\rceil \Big) = 1 - 1 = 0$. Hence $\pshift{1}{3} = \maxR{3} - \minR{3} + \delta_{1,3} = 0 - 0 + 10 = 10$.

\subsection{Holistic Phase Adjustment}

In this section, we describe how we combine the \emph{period shifting} technique and the \emph {phase adjustment} technique. There are three phase types considered in the phased adjustment technique; request phase $\rphase{t}{i}$, start phase $\sphase{t}{i}$, and finish phase $\fphase{t}{i}$. The main difference between our \emph{holistic phase adjustment} and the \emph{phase adjustment} of the Y\&W method is that the phases in our technique can be negative: If phase is negative, it acts like a period shifting.

If $\task{t}$ is a source task, request phase $\rphase{t}{i}$ for each task $\task{i}\not\in\graph{\task{t}}$ is initialized to $-\pshift{t}{i}$, which means that the start time of $\task{t}$ is maximally postponed by preemption of $\task{i}$. In this way, period shifting is merged with phase adjustment so that the request phase can be negative.

If $\task{t}$ is a non-source task, the request phase $\rphase{t}{i}$ depends on the finish phases of predecessors. If it is positive, it means the minimum distance from the release of $\task{t}$ to the next release time of a preempting task $\task{i}$; Fig. \ref{fig:Ex_phase} illustrates the case where task $\task{t}$ has two predecessors $\task{p_1}$ and $\task{p_2}$.

\begin{figure}[ht]
\centerline{\includegraphics[width=8.5cm]{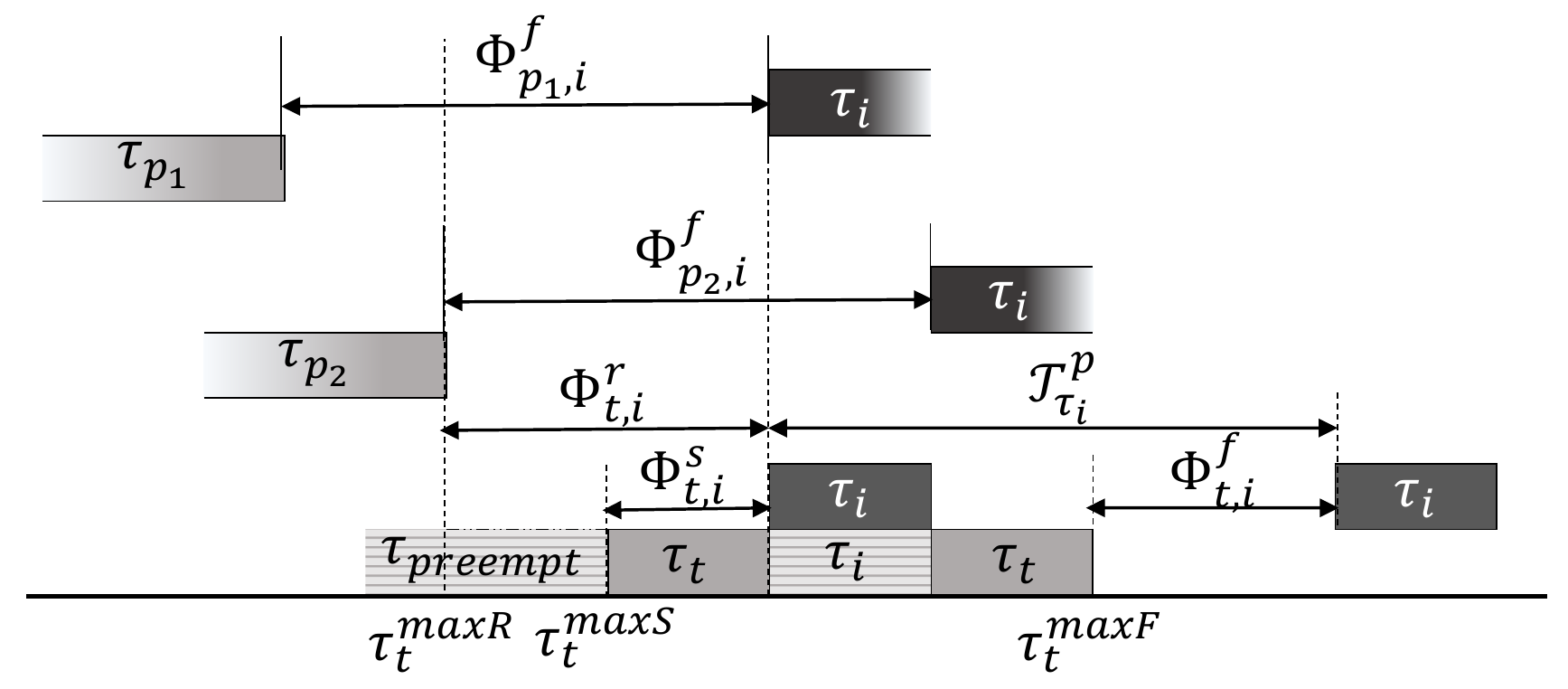}}
\caption{An illustrative example of holistic phase adjustment computation}
\label{fig:Ex_phase}
\end{figure}

Since the phase computation is performed per task independently, the finish phases of the predecessor tasks may be different from each other. When we compute the request phase $\rphase{t}{i}$ and the predecessor tasks see different next request times of $\task{i}$, we have to take the earliest next request time of $\task{i}$ among predecesors for the conservative computation. In Fig. \ref{fig:Ex_phase}, two predecessors $\task{p_1}$ and $\task{p_2}$ of $\task{t}$ see different next request times of $\task{i}$: $\fphase{p_1}{i}+\maxF{p_1}$ and $\fphase{p_2}{i}+\maxF{p_2}$. since $\fphase{p_1}{i}+\maxF{p_1}<\fphase{p_2}{i}+\maxF{p_2}$, we take $\fphase{p_1}{i}+\maxF{p_1}$ as the next request time of $\task{i}$ seen by $\task{t}$, and $\rphase{t}{i}$ becomes $(\fphase{p_1}{i}+\maxF{p_1})-\maxR{t}$, which means the distance from $\maxR{t}$ to the next request time of $\task{i}$.

Note that the inherited phase adjustment can be negative. The negative phase, which acts like a period shifing, cannot be smaller than $-\pshift{t}{i}$. Thus we choose the maximum among the phase adjustment inherited from its predecessors and $-\pshift{t}{i}$. We formulate $\rphase{t}{i}$ as follows;

\begin{equation} \label{eq:Eq7}
\rphase{t}{i}=\left\{\begin{array}{rl}
	-\pshift{t}{i}, & \begin{array}{c}\textrm{if $\task{t}$ is a aource task}\\ \textrm{or $\exists_{\task{p}\in{pred(\task{t})}}{(\proc{p}\not=\proc{t})}$}\end{array}\\
	\max\Big({-\pshift{t}{i}, \min\limits_{\task{p}\in{pred(\task{t})}}{(\fphase{p}{i}+\maxF{p})}-\maxR{t}}\Big), & \textrm{otherwise}\\
\end{array} \right.
\end{equation}

where $pred(\task{t})$ is the immediate predecessor set of task $\task{t}$. In case there is a predecessor mapped to a different processor, $\rphase{t}{i}$ is set to $-\pshift{t}{i}$ for conservative computation.

Based on $\rphase{t}{i}$ and $\maxS{t}$, before computing $\maxF{t}$, the start phase $\sphase{t}{i}$ for each task $\task{i}\not\in\graph{\task{t}}$ is computed as follows:

\begin{equation} \label{eq:Eq8}
\sphase{t}{i}=\left\{\begin{array}{rl}
	(\rphase{t}{i}+\maxR{t})-\maxS{t}, & \textrm{if $\task{i}\not\in\condset{E}{t}$}\\
	\big((\rphase{t}{i}+\maxR{t})-\maxS{t}\big)\bmod\period{\task{i}}, & \textrm{otherwise}\\
\end{array} \right.
\end{equation}

When $\task{i}$ belongs to $\condset{E}{t}$, or $\task{i}\in\condset{E}{t}$, the start phase is made positive by modulo operation to find the distance to the earliest future invocation of $\task{i}$. Otherwise, The start phase $\sphase{t}{i}$ can be negative as $\rphase{t}{i}$.

Similarly, the finish phase $\fphase{t}{i}$ is formulated based on $\sphase{t}{i}$ and $\maxF{t}$ as follows:

\begin{equation} \label{eq:Eq9}
\fphase{t}{i}=\left\{\begin{array}{rl}
	(\sphase{t}{i}+\maxS{t})-\maxF{t}, & \textrm{if $\task{i}\not\in\condset{E}{t}$ or $\proc{t}\in\mathcal{N}$}\\
	\big((\sphase{t}{i}+\maxS{t})-\maxF{t}\big)\bmod\period{\task{i}}, & \textrm{otherwise}\\
\end{array} \right.
\end{equation}

Since there is no preemption during the execution of $\task{t}$ when $\proc{t}\in\mathcal{N}$, finish phase refers to the same invocation of the preempting task as the start phase. The finish phase $\fphase{t}{i}$ is computed after $\maxF{t}$ computation and will be used for the request phases of successors. In the example of Fig. \ref{fig:Ex_phase}, $\rphase{t}{i}$ and $\sphase{t}{i}$ see the same invocation since the difference $(\rphase{t}{i}+\maxR{t})-\maxS{t}$ is yet positive. On the other hand, $\fphase{t}{i}$ sees the next invocation of $\task{i}$, since the invocation of $\task{i}$ seen by $\sphase{t}{i}$ appears during the execution of $\task{t}$. In that case, the time difference from $\maxF{t}$ becomes negative and the modulo operation finds the positive distance to the next invocation.

\begin{lemm} \label{theorem2}
Phase adjustments $\rphase{t}{i}$, $\sphase{t}{i}$, and $\fphase{t}{i}$ are conservative.
\end{lemm}

\begin{proof}
We prove it by induction. First, when task $\task{t}$ is a source task, $\rphase{t}{i}$ is $-\pshift{t}{i}$ and it is conservative according to Lemma \ref{theorem1}. When $\task{i}\not\in\condset{E}{t}$, $\sphase{t}{i}$ and $\fphase{t}{i}$ are computed directly from $\rphase{t}{i}$ by changing the time reference from the release time to the start time ($\maxS{t}$) and ($\maxF{t}$) respectively in equations (\ref{eq:Eq8}) and (\ref{eq:Eq9}). So the consevativeness of the request phase is inherited to the start phase and the finish phase. When $\task{i}\in\condset{E}{t}$, $\sphase{t}{i}$ and $\fphase{t}{i}$ are computed referring to the next release time of $\task{i}$ from $\maxS{t}$ and $\maxF{t}$. Since the worst case preemption scenario to the preempted task is the periodic invocation of the preempting task after the worst case request phase, mod $\period{\task{i}}$ operations in (\ref{eq:Eq8}) and (\ref{eq:Eq9}) find the closest request time of $\task{i}$ from $\maxS{t}$ and $\maxF{t}$ once the request phase is decided. It completes the initial step of the induction process.

Second, we prove the conservativeness of $\rphase{t}{i}$, $\sphase{t}{i}$, and $\fphase{t}{i}$ of non-source task $\task{t}$, assuming that for all $\task{p}\in{pred(\task{t})}$, $\fphase{p}{i}$s are conservative. If $\proc{t}\not=\proc{i}$, then $\rphase{t}{i}$ is $-\pshift{t}{i}$ and it is conservative. Otherwise, we find the closest release time of $\task{i}$ from the finish phase $\fphase{p}{i}$s of predecessors, by (\ref{eq:Eq7}). Hence $\rphase{t}{i}$ is conservative since $\fphase{p}{i}$s are all conservative and the minimum value is chosen. The proof for $\sphase{t}{i}$, and $\fphase{t}{i}$ is similar to the case that task $\task{t}$ is a source task. Q.E.D.
\end{proof}

\section{Optimization Techniques} \label{sec:enhanced_HPA}

In this section, we describe two optimization techniques to tighten the time bounds by removing infeasible preemptions.

\subsection{Exclusion Set Management}

The exclusion technique manages for each task $\task{i}$ a set $\condset{EX}{i}$ which includes tasks that are guaranteed to have no possibility of preempting $\task{i}$. It is obvious that successor tasks belong to this set. If $\task{i}$ always preempts one of the predecessors of $\task{s}$, $\task{s}$ cannot preempt $\task{i}$ since it will always be scheduled after $\task{i}$. In addition, if $\task{j}$ is excluded by $\task{i}$, then all $\task{s}\in\condset{EX}{j}$ are also excluded by $\task{i}$. In summary, the exclusion set $\condset{EX}{i}$ becomes

\begin{equation} \label{eq:Eq21}
\condset{EX}{i}=\{\task{s}|\task{s}\in{descendant(\task{i})}\;or\;\task{i}\in\bigcup_{\task{p}\in{ancestor(\task{s})}}{(\condset{A}{p}\cap\condset{F}{p})}\;or\;\task{s}\in\bigcup_{\task{j}\in\condset{EX}{i}}{\condset{EX}{j}}\}
\end{equation}

where $ancestor(\task{s})$ is a set of ancestors of $\task{s}$ and $descendant(\task{i})$ is a set of descendants of $\task{i}$. Since there is a cyclic dependency in (\ref{eq:Eq21}), iterative computation is required for $\condset{EX}{i}$, initially defined by $descendant(\task{i})$. After time bound computation, it is updated using $\condset{A}{i}$ and $\condset{F}{i}$. Sets $\condset{A}{t}$, $\condset{B}{t}$, $\condset{D}{t}$, $\condset{F}{t}$ and $\condset{G}{t}$ are modified to have an additional condition $\task{s}\not\in\condset{EX}{t}$. It is obvious that the exclusion technique does not affect the conservativeness of the proposed technique.

\subsection{Duplicate Preemption Elimination}

In our baseline technique, preemptions may occur redundantly; Fig. \ref{fig:Ex_redundant} (a) shows an elaborated example that experiences two types of duplicate preemptions. The first type of duplicate preemption may occur between tasks in the same task graph in case a higher priority task has large release time variation. In the scheduling time bound analysis, we detect the preemption possibility by checking if a higher priority task can be released during task execution. In Fig. \ref{fig:Ex_redundant} (a), $\task{5}$ can preempt both $\task{4}$ and $\task{8}$ because its release time varies between 20 and 75.

The second type of duplicate preemption occurs between tasks in different task graphs. In the phase adjustment technique, it is assumed that a preempting task preempts a predecessor task first. In the example of Fig. \ref{fig:Ex_redundant}, $\task{0}$ preempts $\task{4}$ and phase adjustment is performed afterwards. The request phase of $\task{7}$ to $\task{0}$ is reset to $-\pshift{7}{0}$ since $\task{7}$ and $\task{0}$ are assigned to different processors, according to equation (\ref{eq:Eq7}). Since the request phase of $\task{8}$ is inherited from $\task{7}$, $\task{8}$ experiences another preemption by $\task{0}$.

As a result, the WCRT is overestimated as illustrated in Fig. \ref{fig:Ex_redundant} (a) that contains both types of duplicate preemptions.

\begin{figure}[ht]
\centerline{\includegraphics[width=11cm]{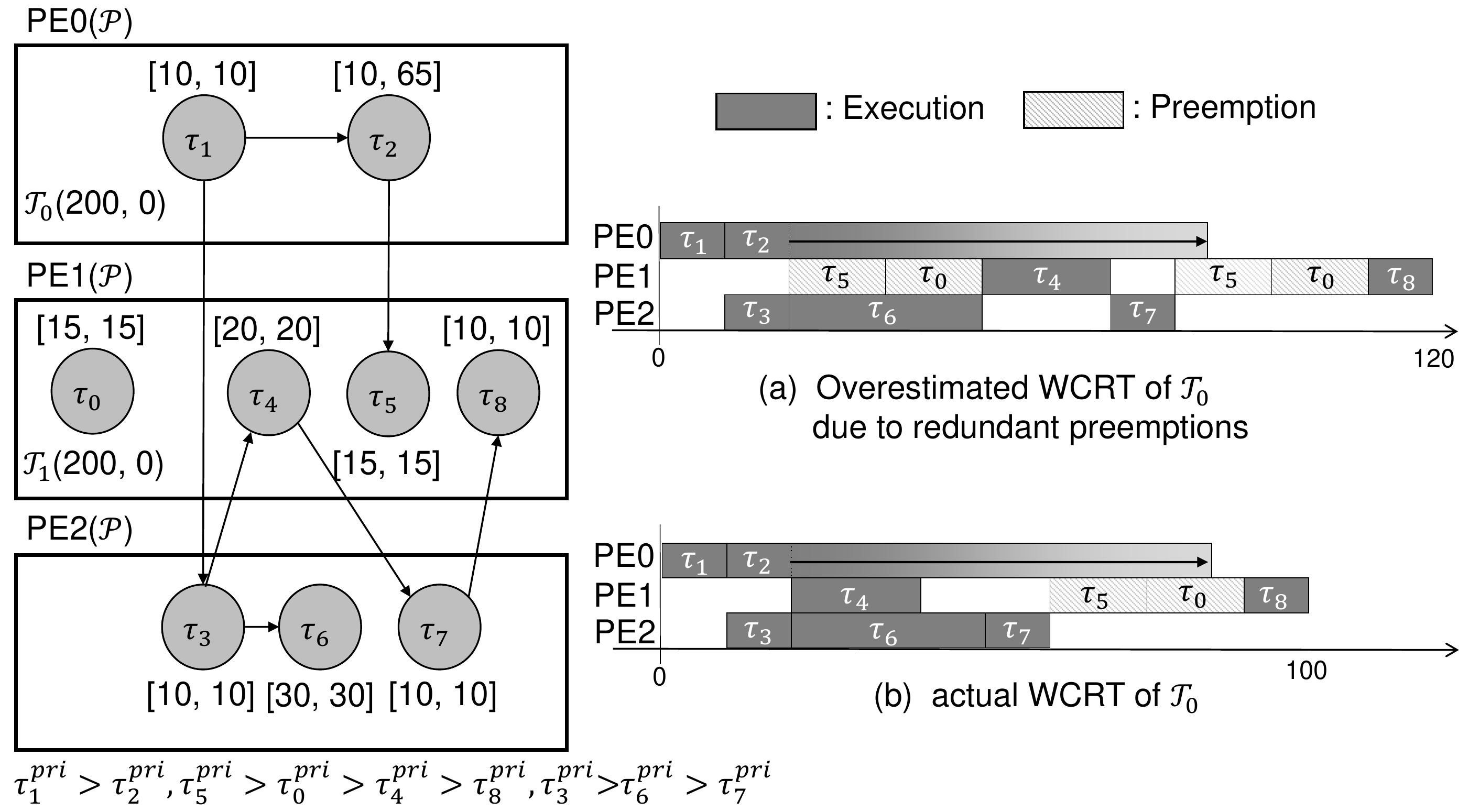}}
\caption{An example of common preemption elimination}
\label{fig:Ex_redundant}
\end{figure}

\IncMargin{1em}      
\begin{algorithm}
\SetAlgoNoLine
\DontPrintSemicolon
\SetKwProg{myProc}{Procedure}{}{}

\SetKwFunction{RemoveDP}{RemoveDP}
\myProc{\RemoveDP{$\task{t}$}}{
	$T_t^r=T_t^s=0$\;
	\Repeat{$T_t^s$ or $T_t^r$ is changed}{
		$T_t^s=0$\;
		$T_t^r=RecursiveRemoveDP(\task{t},\task{t},\phi)$\;
	}
}

\SetKwFunction{RecursiveRemoveDP}{RecursiveRemoveDP}
\myProc{\RecursiveRemoveDP{$\task{c},\task{t},\condset{I}{t}$}}{
	find $\task{cri1}$ that satistifes $(\task{cri1}\in{pred(\task{c})}, \maxF{cri1}=\max_{\task{p}\in{pred(\task{c})}}{\maxF{p}})$\;
	find $\task{cri2}$ that satistifes $(\task{cri2}\in{pred(\task{c})}, \maxF{cri2}=\max_{\task{p}\in{pred(\task{c})-\{\task{cri1}\}}}{\maxF{p}})$\;
	$T_r=T_e=0$\;
	\If{$\task{c}\not=\task{t}$}{
		$T_e$ is increased by $\sum_{\task{s}\in\condset{J}{c}}{\WCET{s}}$\;
		$\condset{I}{t}=\condset{I}{t}\cup\condset{J}{c}$\;
		\If{$\exists_{\task{p}}{(\task{p}\in{pred(\task{t})},\maxF{p}=\maxR{t},detect_p)}$}{
			$T_e$ and $T_t^s$ are increased by $\sum_{\task{s}\in\condset{K}{c}}{PC_{c,s}\cdot\WCET{s}}$\;
			$\condset{I}{t}=\condset{I}{t}\cup\condset{K}{c}$\;
		}
	}
	\If{$\task{c}$ is a source task} {
		\KwRet $T_e$
	}
	$T_r=RecursiveRemoveDP(\task{cri1},\task{t},\condset{I}{t})$\;
	\If{$\task{c}\not=\task{t}$} {
		$T_r=\max\big(0, T_r-\sum_{\task{s}\in\condset{L}{c}}{\min(\WCET{s},\maxF{s}-(\maxR{c}-T_r)}\big)$\;
	}
	\eIf{$\task{cri2}$ is not found} {
		\KwRet $T_r+T_e$
	}{
		\KwRet $\min(T_r, \maxF{cri1}-\maxF{cri2})+T_e)$
	}
}
\caption{Algorithm of duplicate preemption elimination}
\label{alg:commonpreemption}
\end{algorithm}
\DecMargin{1em}

To avoid duplicate preemptions, we devised an optimization that traces back the schedule and moves duplicate preemptions from ancestors to a target task. The proposed optimization heuristic is based on the abstruse fact that later preemption gives worse response time than earlier preemption when there are duplicate preemptions, which is stated in Theorem \ref{theorem8} below.

Algorithm \ref{alg:commonpreemption} presents the psuedo code of the proposed heuristic. It is invoked by $RemoveDP(\task{t})$ (lines 1-7) where $\task{t}$ is the task whose schedule time bound is computed. $T_t^r$ represents how much the release time of $\task{t}$ is reduced by removing duplicate preemptions, and $T_t^s$ represents how much preemption delay caused by other task graphs should be moved from predecessors to $\task{t}$. We recursively trace back the schedule of critical path (lines 9 and 22) where $\task{c}$ and $\task{cri1}$ represent currently visited task and the parent task of $\task{c}$ on the critical path, respectively. We initialize $T_t^r$ to zero and repeat $RecursiveRemoveDP$ until $T_t^r$ is converged.

In algorithm \ref{alg:commonpreemption}, $PC_{t,i}$ means the preemption count from $\task{i}$ to $\task{t}$. Boolean flag $detect_t$ checks if the release time of task $\task{t}$ cannot be reduced as much as the removed preemption time due to a predecessor mapped to a different PE. The boolean flag is inherited to the successors.

Let the currently visited task be $\task{c}$ in the recursive call, $RecursiveRemoveDP$. At first, we compute the amount of duplicate preemptions from tasks in the same task graph which can preempts both $\task{c}$ and $\task{t}$ (lines 13-14), where those tasks belong to set $\condset{J}{c}=\{\task{s}|\task{s}\in\condset{D}{c}\cup\condset{G}{c},\task{s}\not\in\condset{I}{t},\task{s}\not\in\condset{EX}{t},\proc{s}=\proc{t},\pri{s}>\pri{t},\maxR{t}-T_t^r<\maxF{s}\}$. Even though preemption from the same task can be seen several times on the critical path, we need to consider only the recent preemption in this recursive function. Thus we manage task set $\condset{I}{t}$ while traversing the critical path in order to consider only the recent preemption. If $detect_p$ value for $\task{t}$ is true, we also move the duplicate preemptions from tasks in the other task graphs to $\task{t}$ (lines 15-18), where those tasks belong to set $\condset{K}{c}=\{\task{s}|\task{s}\not\in\graph{\task{t}},PC_{c,s}\not=0,\task{s}\not\in\condset{I}{t},\proc{s}=\proc{t},\pri{s}>\pri{t}\}$. After computation of the amount of duplicate preemptions, $RecursiveRemoveDP$ is called recursively for the critical path predecessor $\task{cri1}$, and we get the returned value as $T_r$ (line 22), which is the possible release time reduction of $\task{c}$. However, there can be additional preemptions if $\task{c}$ is released at $\maxR{c}-T_r$. We conservatively find the amount of additional preemptions and reduce $T_r$ (line 23), where a task that incurs additional preemption belongs to set $\condset{L}{c}=\{\task{s}|\task{s}\not\in\condset{D}{c}\cup\condset{G}{c},\task{s}\in\graph{\task{c}},\task{s}\not\in\condset{EX}{c},\task{s}\not\in{ancestor(\task{c})},\proc{s}=\proc{c},\pri{s}>\pri{c},\minS{s}\le\maxR{c},\maxR{c}-T_r<\maxF{S}\}$. For the non-preemptive scheduling policy, the condition $\pri{s}>\pri{c}$ is removed in $\condset{L}{c}$. Note that release time reduction can be bounded by the other predecessors. If there is more than two predecessors of $\task{c}$, we bound the reduced release time with the second largest finish time among the predecessors (lines 25-26). Note that $T_t^r$ is used in the formula of $\condset{J}{c}$. Hence we set $T_t^r$ to zero initially and repeat $RemoveDP$ until $T_t^r$ is converged.

Refer to the example in Fig. \ref{fig:Ex_redundant}. For task $\task{8}$, we trace back the schedule of tasks $\task{7}$, $\task{4}$, $\task{3}$, and $\task{1}$ in order when we call $RemoveDP(\task{8})$. When $\task{4}$ is visited, We find out that $\task{0}\in\condset{K}{4}$ and $\task{5}\in\condset{J}{4}$, and both task can preempt $\task{8}$. We remove those duplicate preemptions, then the maximum release time of $\task{7}$ is reduced to $\maxR{7}-T_r=40$. When returning to $\task{7}$, We can know that $\task{6}\in\condset{L}{7}$ so that the start time of $\task{7}$ cannot be earlier than the maximum finish time of $\task{6}$, which is $\maxR{7}-(T_r-\maxF{6}-(\maxR{7}-T_r ))=50$. Finally, we know that $\task{8}$ can be released at $\maxR{8}-T_t^r=60$ after removing all duplicate preemptions of its predecessors, which is the actual worst case.

The HPA equations need to be modified after Algorithm \ref{alg:commonpreemption} is applied. We call $RemoveDP(\task{t})$ after $\maxR{t}$ computation and set the reduced maximum release time $\hat\tau_t^{maxR}$ to $\maxR{t}-T_t^r$. Then $\maxR{t}$ is replaced by $\hat\tau_t^{maxR}$ in equations (\ref{eq:Eq13}),(\ref{eq:Eq14}), and sets $\condset{B}{t}$ and $\condset{D}{t}$. The terms $\maxF{s}-\maxR{t}$ and $\maxS{t}-\maxR{t}+1-\rphase{t}{s}$ in the formula of $Delay_t^h$ is changed to $\maxF{s}-\hat\tau_t^{maxR}$ and $\max⁡(0,\maxS{t}-\maxR{t})+1-\rphase{t}{s}$ respectively. And the request phase is not reset to the period shifting value in equation (\ref{eq:Eq7}) when there is a predecessor mapped to different processors. The sets $\condset{B}{t}$ and $\condset{D}{t}$ have an additional condition $\task{s}\not\in{ancestor(\task{t})}$ since $\hat\tau_t^{maxR}$ can be smaller than the finish times of predecessors. $\maxS{t}$ is changed to have the initial value of $\hat\tau_t^{maxR}$ and is bounded by $\maxR{t}$. And $\maxS{t}$ is changed to add $T_t^s$ that is the sum of removed preemptions from the tasks in the other task graphs as follows:

\begin{equation} \label{eq:Eq23}
\maxS{t}=\max(\maxR{t},\hat\tau_t^{maxR}+T_t^s+Delay_t^l+Delay_t^h)
\end{equation}

\begin{theo} \label{theorem8}
Algorithm \ref{alg:commonpreemption} that removes the duplicate preemptions from the ancestors in the critical path preserves the conservativeness of the schedule time bound.
\end{theo}

\begin{proof}
Refer to the electronic appendix for the proof.
\end{proof}

\section{Overall HPA Algorithm} \label{sec:overall_alg}

Now we ready to summarize the overall algorithm of the proposed technique. We compute the schedule time bounds of tasks and phase adjustment values until all time bounds are converged. The algorithm \ref{alg:overallalg} shows the outermost iterative routine that integrates all computations.

\IncMargin{1em}      
\begin{algorithm}
\SetAlgoNoLine
\DontPrintSemicolon
\SetKwProg{myProc}{Procedure}{}{}

\SetKwFunction{HPA}{HPA}
\myProc{\HPA}{
	period shifting $\pshift{i}{j}=\jitter{\task{j}}$ for every task pair $(\task{i}, \task{j})$\;
	exclusion set $\condset{EX}{i}=descendant(\task{i})$ for every task $\task{i}$\;
	\Repeat{any value is changed and $\forall_{\graph{}}{(\mathcal{R}_{\graph{}}\le\deadline{})}$}{
		\ForEach{task $\task{t}$ in topological and priority descending order}{
			compute $\minR{t}$, $\minS{t}$, and $\minF{t}$\;
			compute $\maxR{t}$, $\hat\tau_t^{maxR}$, $\forall_{\task{i}\not\in\graph{\task{t}}}{\rphase{t}{i}}$, $\maxS{t}$, $\forall_{\task{i}\not\in\graph{\task{t}}}{\sphase{t}{i}}$, $\maxF{t}$, and $\forall_{\task{i}\not\in\graph{\task{t}}}{\fphase{t}{i}}$\;
		}
		update $\pshift{i}{j}$ for every task pair $(\task{i}, \task{j})$\;
		update $\condset{EX}{i}$ for every task $\task{i}$\;
	}
}
\caption{HPA overall algorithm}
\label{alg:overallalg}
\end{algorithm}
\DecMargin{1em}

At first, period shifting and exclusion set are initialized (lines 2-3). Then we iteratively compute the time bounds of tasks (lines 5-8). Tasks are visited according to the topological order first and priority descending order among independent tasks. Each time bound and phase is computed in the written order. After the time bound computation, period shifting and exclusion set set are updated (lines 9-10). This process is repeated until every value is converged. If there is a task graph that violates its deadline, we stop the iteration since it is not schedulable (line 11).

\section{Experiments} \label{sec:experiment}

As the reference RTA technique, we implemented the Y\&W method following the pseudo code in [Yen and Wolf 1998]\ignore{\cite{ref3}}. For comparison with MAST and SymTA/S, we use available tools; MAST suite [Harbour 2001]\ignore{\cite{ref16}} and pyCPA [Diemer and Axer 2012]\ignore{\cite{ref17}} that is a freely available compositional performance analysis tool similar to SymTA/S. The actual WCRT, which is denoted as Optimal, was obtained by an ILP-based approach [Kim et al. 2012]\ignore{\cite{ref2}}. The proposed HPA technique is available on-line [Choi et al. 2014]\ignore{\cite{ref18}}. 

\begin{table}[ht]
\tbl{WCRT estimation results for simple examples in Fig. \ref{fig:Ex_underestimation1}, \ref{fig:Ex_underestimation2}, \ref{fig:Ex_overestimation1}, \ref{fig:Ex_overestimation2}, and \ref{fig:Ex_redundant} \label{tbl:table1}}{
\begin{tabular}{|c|c|c|c|c|c|}
\hline
& HPA & Y\&W & MAST & pyCPA & Optimal \\
\hline
$\graph{0}$ in Fig. \ref{fig:Ex_underestimation1} & 40 & 35 & 40 & 50 & 40 \\
\hline
$\graph{0}$ in Fig. \ref{fig:Ex_underestimation2} & 130 & 110 & 130 & 270 & 130 \\
\hline
$\graph{0}$ in Fig. \ref{fig:Ex_overestimation1} & 140 & 150 & 140 & 300 & 140 \\
\hline
$\graph{0}$ in Fig. \ref{fig:Ex_overestimation2} & 70 & 100 & $\times$ & 210 & 70 \\
\hline
$\graph{0}$ in Fig. \ref{fig:Ex_redundant} & 100 & 120 & $\times$ & 310 & 100 \\
\hline
\end{tabular}}
\end{table}

Table \ref{tbl:table1} shows the comparison results for the examples shown in this paper. As discussed in Section \ref{sec:YW_review}, the Y\&W method fails to find the WCRT for the examples of Fig. \ref{fig:Ex_underestimation1} and Fig. \ref{fig:Ex_underestimation2}. MAST gives tight WCRT results for linear graphs, but no result for Fig. \ref{fig:Ex_overestimation2} and Fig. \ref{fig:Ex_redundant} since it supports only chain-structured graphs, and pyCPA provides highly overestimated results for all examples, since it pessimistically uses response time analysis ignoring the task dependency between processing elements. Note that the proposed HPA technique gives the optimal WCRT results for all these examples.

For extensive comparison, we generate graphs randomly; the number of task graphs varies from 3 to 5, the number of total tasks from 30 to 50, and the number of processing elements from 3 to 5. The $\BCET{}$ and $\WCET{}$ of each task are randomly selected in the range of [500, 1000] and [$\BCET{}$,$\BCET{}\times1.5$] respectively. The period and jitter of task graphs are randomly chosen but repaired to be schedulable if needed. Note that the problem size is too big to find optimal WCRTs with an ILP-based approach.

\begin{table}[ht]
\tbl{Comparison with the Y\&W method and pyCPA for 100 random examples \label{tbl:table2}}{
\begin{tabular}{|c|c|c|c|c|c|c|}
\hline
& Win & Tie & Lose & Max{\%} & Min(\%) & Avg(\%) \\
\hline
vs Y\&W & 193 & 200 & 1 & 39.42 & -0.69 & 3.86 \\
\hline
vs pyCPA & 394 & 0 & 0 & 359.00 & 16.85 & 157.60 \\
\hline
\end{tabular}}
\end{table}

Table \ref{tbl:table2} shows the comparison results with the Y\&W method and pyCPA. For the comparison with the Y\&W method, we assume that all processing elements use the preemptive scheduling policy, and there is no jitter of input arrival. Since tasks may have multiple predecessors or successors, MAST is excluded in this comparison. The total number of task graphs is 394 in 100 randomly generated examples. The first column shows how many cases the HPA technique produces a tighter bound than the other methods. Similarly, the second and the third columns show how many cases HPA technique produces an equivalent bound and a looser bound. Columns Max, Min, and Avg indicate the maximum, minimum, and average WCRT estimation gaps between HPA and the other approaches. HPA gives tighter bounds than the Y\&W method for almost half of task graphs. A looser bound than the Y\&W method was found for one task graph only with very small estimation gap, 0.69\%. pyCPA provides highly overestimated WCRT results for all task graphs, which are on average 2.57 times larger than the results of HPA.

\begin{table}[ht]
\tbl{Comparison with MAST and pyCPA for 100 random examples \label{tbl:table3}}{
\begin{tabular}{|c|c|c|c|c|c|c|}
\hline
& Win & Tie & Lose & Max{\%} & Min(\%) & Avg(\%) \\
\hline
vs MAST & 393 & 8 & 6 & 104.31 & -4.48 & 18.95 \\
\hline
vs pyCPA & 406 & 1 & 0 & 295.16 & 0.00 & 44.83 \\
\hline
\end{tabular}}
\end{table}

For the comparison with MAST, we create another 100 random examples that are restricted to chain-structured graphs. In contrast to the previous experiment, random examples can have arbitrary mixture of preemptive and non-preemptive processing elements, and the jitter of input arrival is allowed. The number of total task graphs is 407. As shown in Table \ref{tbl:table3}, HPA shows remarkable performance advantage over the other methods. HPA provides tighter bounds than MAST and pyCPA in most cases.

Experiments are futher conducted to examine how the performance of the proposed HPA technique scales as the number of tasks, the number of graphs, and the execution time variation increases, compared with the other methods. The results are depicted in Fig. \ref{fig:Exp_taskvariation}, \ref{fig:Exp_graphvariation}, and \ref{fig:Exp_execvariation}, respectively. In these experiments, the default ranges of the number of processing elements, the number of task graphs, and the number of total tasks are set to [5,10], [5,7], and [30,50] respectively. For the left graphs in three figures, random examples are generated to have no jitter and only preemptive processing elements to compare with the Y\&W method and pyCPA. For the right graphs, the task graphs in random examples are restricted to chain-structed graph but non-preemptive processing elements are allowed to compare with MAST and pyCPA. For all graphs in three figures, right y axis indicates the average WCRT estimation gap between HPA and pyCPA. For left(right) graphs, left y axis indicates the average WCRT estimation gap between HPA and the Y\&W method(MAST). for each data point, the average value is obtained from 100 random examples.

\begin{figure}[ht]
\centerline{\includegraphics[width=14.2cm]{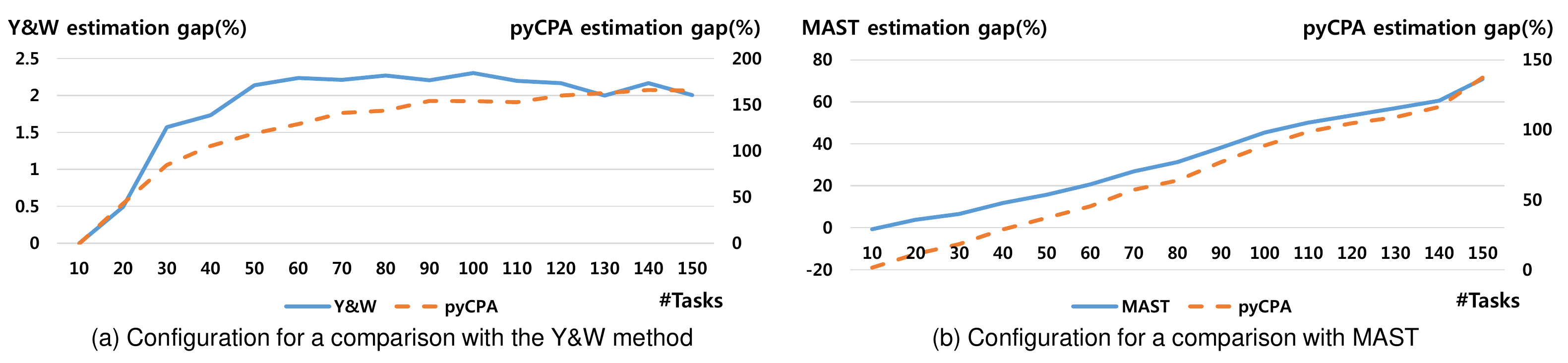}}
\caption{WCRT estimation gap increases as the number of total tasks increases}
\label{fig:Exp_taskvariation}
\end{figure}

Fig. \ref{fig:Exp_taskvariation} shows the change of estimation gap while varing the number of total tasks from 10 to 100. The estimation gap increases by the number of total tasks. For the right graph that uses chain-structed graph, the estimation gap is increased linearly, which confirms that HPA handles the effect of dependency effectively.

\begin{figure}[ht]
\centerline{\includegraphics[width=14cm]{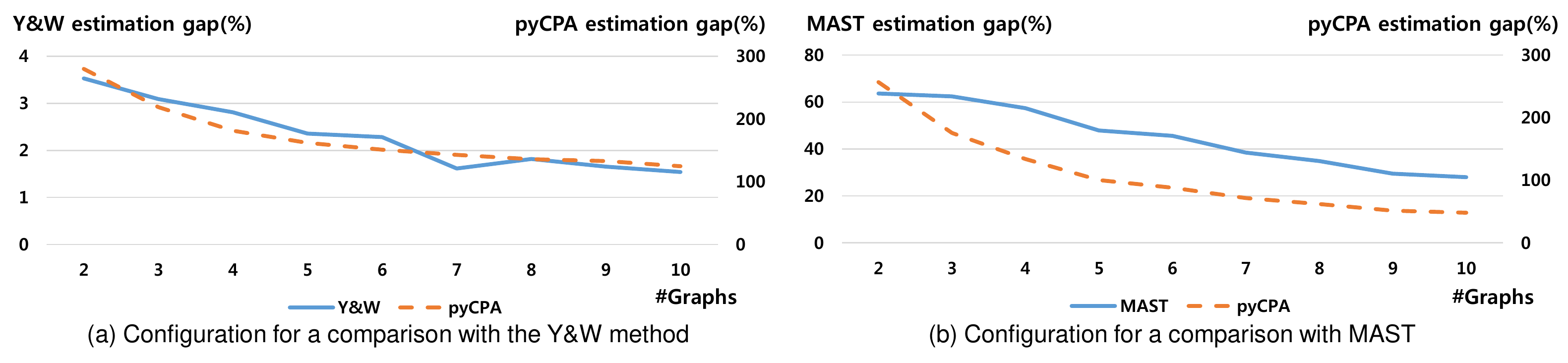}}
\caption{WCRT estimation gap decreases as the number of task graphs increases with the same number of tasks in total}
\label{fig:Exp_graphvariation}
\end{figure}

Fig. \ref{fig:Exp_graphvariation} shows the result of experiment that varies the number of task graph from 2 to 10 while fixing the number of total tasks to 100. In this experiment, the estimation gap decreases. It is because that as the task dependency decreases as well, the inter-application interference considered by the RTA analysis becomes dominant over intra-application interference that is computed by the schedule time bound analysis.

\begin{figure}[ht]
\centerline{\includegraphics[width=13.8cm]{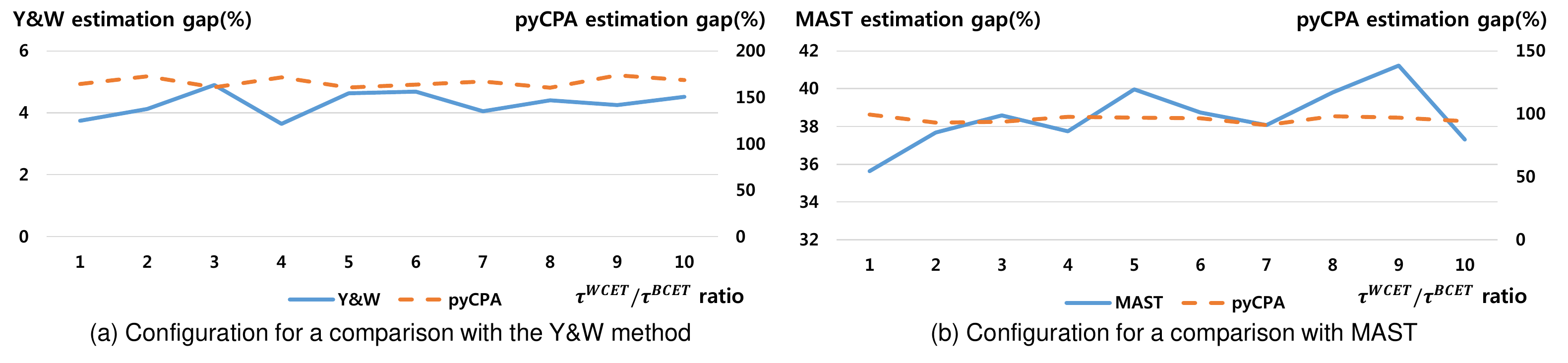}}
\caption{WCRT estimation gap change graphs with variation of $\frac{\WCET{}}{\BCET{}}$ ratio}
\label{fig:Exp_execvariation}
\end{figure}

The performance variation over the $\frac{\WCET{}}{\BCET{}}$ ratio from 1 to 10 is shown in Fig. \ref{fig:Exp_execvariation}. The figure shows that the execution time variation does not incurs any meaningful change in the performance gap.

\begin{figure}[ht]
\centerline{\includegraphics[width=13.7cm]{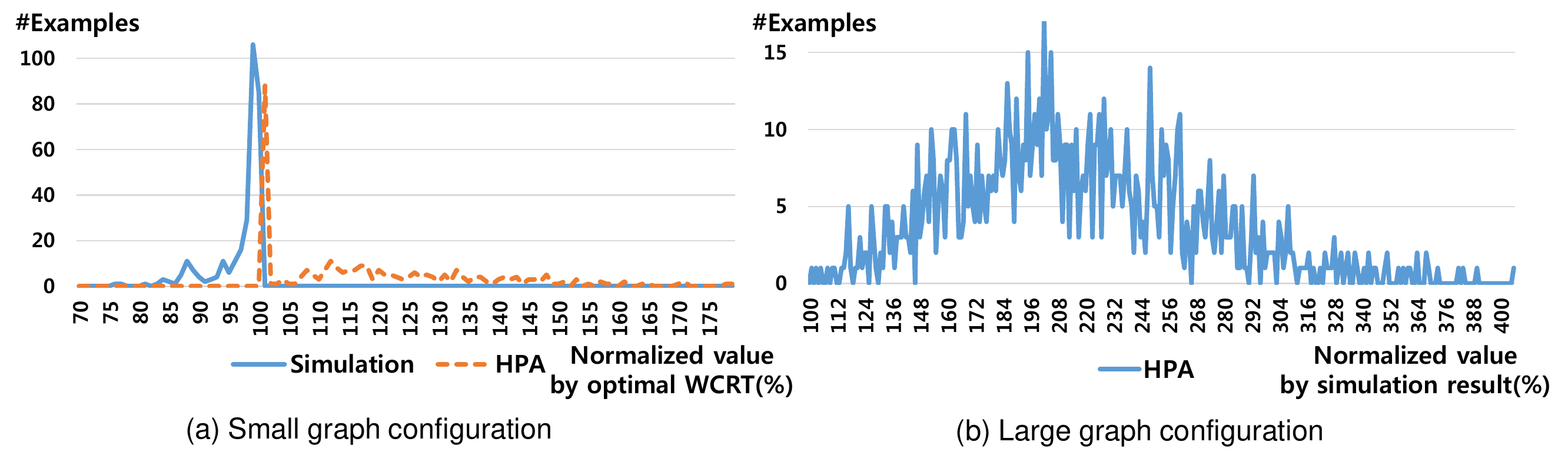}}
\caption{Distribution graphs of the estimated WCRT from the HPA}
\label{fig:Exp_distribution}
\end{figure}

Fig. \ref{fig:Exp_distribution} shows the distribution of the estimated WCRTs from HPA with 300 randomly generated examples that have arbitrary mixture of preemptive and non-preemptive scheduling policies and no restriction on the graph topology. Fig. \ref{fig:Exp_distribution} (a) shows the distribution with small size examples and Fig. \ref{fig:Exp_distribution} (b) with large size examples. For small(large) size examples, the ranges of the number of processing elements, the number of task graphs, and the number of total tasks are set to [2,3]([5,10]), [1,3]([5,7]), and [10,15](60,100), respectively. We also performed Monte-Carlo simulation to obtain the WCRT empirically by sampling graph instances 200 million times for each example. Optimal WCRTs are found with an ILP-based approach only for small size examples. Fig. \ref{fig:Exp_distribution} (a) depicts the distribution of the WCRTs obtained from Monte-Carlo simulation and HPA, normalized by optimal WCRTs. The x axis indicates the normalized value (\%) and the y axis presents the number of examples. Note that both Monte-Carlo simulation and HPA provide results close to the actual WCRT in majority cases. But Monte-Carlo simulation may not find the true WCRT despite 200 million times of simulation, which confirms the need of conservative estimation techniques. For large size examples, the distribution results from HPA are normalized by the near-WCRTs obtained from Monte-Carlo simulation in Fig. \ref{fig:Exp_distribution} (b) since optimal WCRTs cannot be obtained. Since Monte-Carlo simulation gives underestimated WCRTs, the amount of overestimation might be quite exaggerated. It shows that the estimated WCRT from HPA may have about 100\% overestimation on average.

\begin{figure}[ht]
\centerline{\includegraphics[width=10cm]{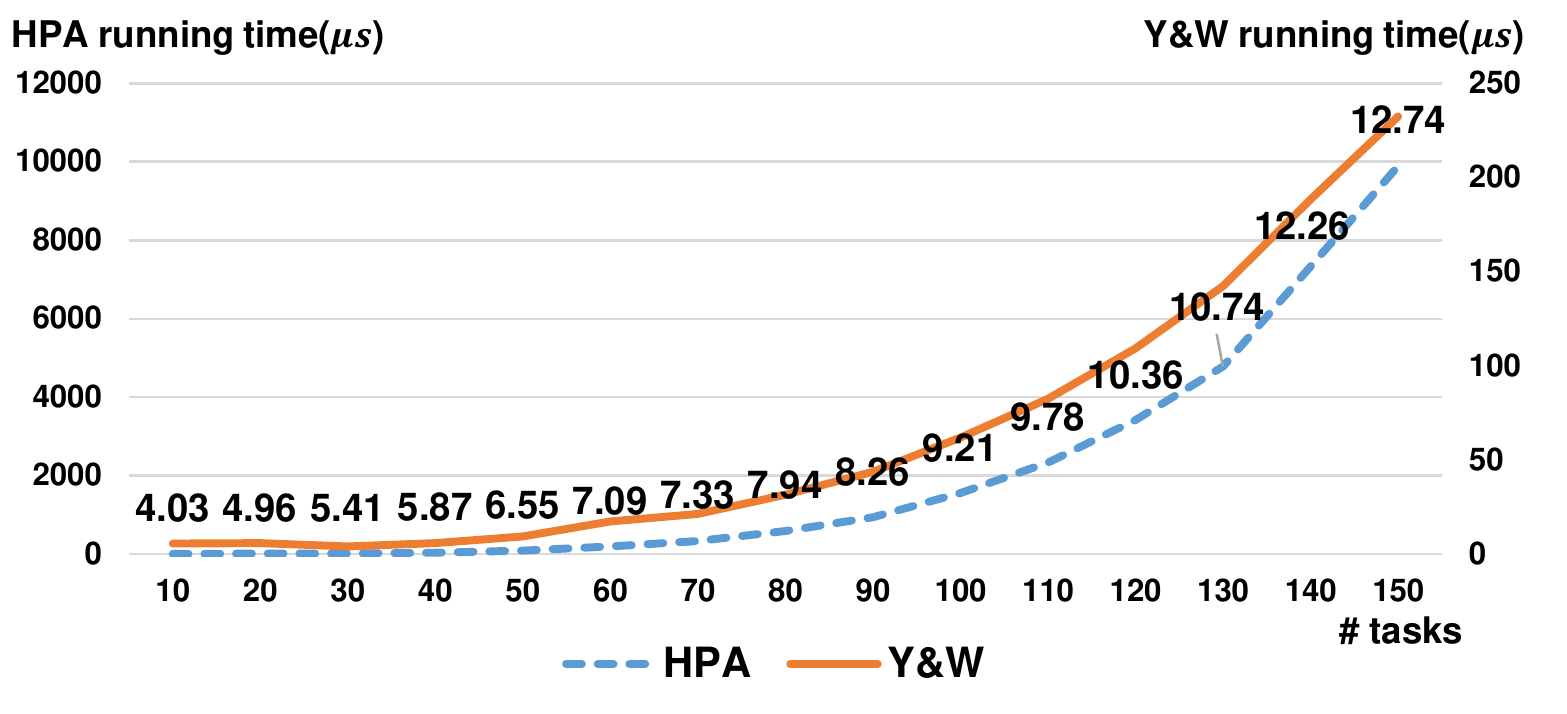}}
\caption{Average running time graph with average number of iteration}
\label{fig:Exp_performance}
\end{figure}

Finally, we measure the running time of HPA on 3.4 GHz i7 machine with 8GB main memory and the iteration number of convergence to verify the scalability of the proposed technique. We vary the number of total tasks from 10 to 150 and use 100 randomly generated examples for each number of total tasks. The average values are plotted in Fig. \ref{fig:Exp_performance}, where the left y-axis and the right y-axis indicate the running time of HPA and Y\&W in micro seconds, respectively. The numbers labeled to the graph of the HPA indicate the average numbers of iterations for convergence. Every example was converged within maximum 20 iterations. Although HPA is slower than the Y\&W method by an order of magnitude, both methods show similar scalability as shown in Fig. \ref{fig:Exp_performance}.

\section{Conclusion} \label{sec:conclusion}

In this paper, we addressed a very challenging problem that is to tightly estimate the worst case response time of an application in a distributed embedded system. It is shown that a state-of-the-art technique, Y\&W method, fails to find a conservative WCRT bound. Thus we propose a hybrid performance analysis (HPA) method that combines the scheduling time bound analysis and the response time analysis to consider inter-task interference between different tasks. It finds a conservative and tight WCRT bound, considering task dependency, execution time variation, arbitrary mixture of fixed-priority preemptive and non-preemptive processing elements, and input jitters. Experimental results show that it produces tighter bounds than the Y\&W method, MAST, and pyCPA. Convergence and scalability of the proposed technique are confirmed empirically. Since it is an iterative technique with tens of non-linear formulas, it is a very difficult problem to prove the convergence formally, which is left as a future work.



\appendixhead{CHOI}


\bibliographystyle{ACM-Reference-Format-Journals}

\received{February 2007}{March 2009}{June 2009}

\elecappendix

\newtheorem{APPtheorem}{Theorem}[section]
\newtheorem{APPlemma}{Lemma}
\newtheorem{APPdefinition}{Definition}

\renewcommand{\theAPPtheorem}{{\thesection}.\arabic{APPtheorem}}
\renewcommand{\theAPPlemma}{{\thesection}.\arabic{APPlemma}}
\renewcommand{\theAPPdefinition}{{\thesection}.\arabic{APPdefinition}}

\medskip

\section{Proofs for Optimization Techniques}

We prove that the optimization technique of duplicate preemption elimination preserves the conservativeness of the HPA technique. It is a rather long proof. We will make several definitions, lemmas, and theorems in this section.

\begin{APPdefinition} \label{theorem4_2}
$pe(\task{t})[x,y]$ is defined as the sum of execution time of tasks of which priority is higher than $\task{t}$ from time $x$ to time $y$.
Formally,
$$pe(\task{t})[x,y] = \sum_{\task{s}\in{S_{pe(\task{t})[x,y]}}}{\min(y,\finish{s})-\max(x,\start{s})}$$
where $S_{pe(\task{t})[x,y]}=\{\task{s}|\pri{s}>\pri{t}, (x \le \start{s}<y ~or~ x< \finish{s} \le y)\}$ for preemptive scheduling and
$S_{pe(\task{t})[x,y]}=\{\task{s}|(\pri{s}>\pri{t} ~or~ \start{s} < \start{t} ), (x \le \start{s}<y ~or~ x< \finish{s} \le y)\}$ for nonpreemptive scheduling.
\end{APPdefinition}

The following two lemmas hold and their proofs are trivial.

\begin{APPlemma} \label{theorem4_3}
$0 \le pe(\task{t})[x,y] \le y-x.$
\end{APPlemma}

\begin{APPlemma} \label{theorem4_4}
$pe(\task{t})[x,z] = pe(\task{t})[x,y] + pe(\task{t})[y,z].$
\end{APPlemma}

\ignore{
\begin{proof}
pe(\task{t})[x,y] + pe(\task{t})[y,z] \\
= \sum_{\pri{s}>\pri{t} ~and ~(x \le \start{s}~{\text {or}}~ \finish{s} \le y)}{min(y,\finish{s})-max(x,\start{s})}
+ \sum_{\pri{s}>\pri{t} ~and ~(y \le \start{s}~{\text {or}}~ \finish{s} \le z)}{min(z,\finish{s})-max(y,\start{s})} \\
= \sum_{\pri{s}>\pri{t} ~and ~(x \le \start{s}~{\text {or}}~ \finish{s} \le y ~or y \le \start{s}~{\text {or}}~ \finish{s} \le z)} {min(y,\finish{s})-max(x,\start{s})}+{min(z,\finish{s})-max(y,\start{s})}
...
\end{proof}
}

\begin{APPtheorem} \label{theorem4_5}
for all $\task{t}$, if $\release{t} \le \maxR{t}-\Delta$
then it always holds that $\finish{t} \le \maxF{t} - \Delta + pe(\task{t})[\maxR{t}-\Delta,\maxR{t}]$.
\end{APPtheorem}

\begin{proof}
We prove it by contradiction. Suppose that $\finish{t}$ exists such that
$\finish{t} > \maxF{t} - \Delta + pe(\task{t})[\release{t},\finish{t}]$ where $\release{t} \le \maxR{t}-\Delta$.
It is obvious that $\maxF{t} \ge \maxR{t} + \WCET{t} + pe(\task{t})[\maxR{t},\maxF{t}]$ and $\finish{t} \le \release{t} + \WCET{t} + pe(\task{t})[\release{t},\finish{t}]$, by Definition \ref{theorem4_2}.
$$\finish{t} > \maxF{t} - \Delta + pe(\task{t})[\maxR{t}-\Delta,\maxR{t}]$$
$$\ge \maxR{t} + \WCET{t} + pe(\task{t})[\maxR{t},\maxF{t}] - \Delta + pe(\task{t})[\maxR{t}-\Delta,\maxR{t}]$$
$$= \maxR{t} + \WCET{t} - \Delta + pe(\task{t})[\maxR{t}-\Delta,\maxF{t}].$$
Thus,
$$\release{t} + \WCET{t} + pe(\task{t})[\release{t},\finish{t}]
> \maxR{t} + \WCET{t} - \Delta + pe(\task{t})[\maxR{t}-\Delta,\maxF{t}].$$
$$\release{t} + pe(\task{t})[\release{t},\maxR{t}-\Delta] + pe(\task{t})[\maxR{t}-\Delta,\finish{t}]$$
$$> \maxR{t} - \Delta + pe(\task{t})[\maxR{t}-\Delta,\maxF{t}].$$

$\release{t} + pe(\task{t})[\release{t},\maxR{t}-\Delta] > \maxR{t} - \Delta $
since $pe(\task{t})[\maxR{t}-\Delta,\maxF{t}] \ge pe(\task{t})[\maxR{t}-\Delta,\finish{t}]$.\\

$ \maxR{t} - \Delta < \release{t} + pe(\task{t})[\release{t},\maxR{t}-\Delta] \le \maxR{t} -\Delta$
since $pe(\task{t})[\release{t},\maxR{t}-\Delta] \le \maxR{t}-\Delta - \release{t}$.\\

It is a contradiction. Hence $\finish{t} \le \maxF{t} - \Delta + pe(\task{t})[\maxR{t}-\Delta,\maxR{t}]$. Q.E.D.
\end{proof}

Theorem~\ref{theorem4_5} explains how much the maximum finish time can be reduced
if the maximum release time decreases. Since the release time depends on the finish times of
predecessor tasks, the theorem presents the effect of the finish time of predecessor tasks onto
the finish time of the target task. This theorem will be generalized later in theorem~\ref{theorem4_12}.

Hereafter, $pe(\task{t},\Delta)$ denotes $pe(\task{t})[\maxR{t}-\Delta,\maxR{t}]$ for brevity.

It is trivial that if a predecessor task always finishes early then the target task is released early.
Therefore, for $\task{s} \in pred(\task{t})$, if $\forall_{\finish{s}}{(\finish{s} \le \maxF{s}-\Delta_{\task{s}})}$
then $\release{t} \le \max_{\task{s} \in pred(\task{t})}{(\maxF{s}-\Delta_{\task{s}})}$.

\begin{APPdefinition} \label{theorem4_8}
$ancestor(\task{t})$ is a set of ancestors of $\task{t}$, formally $ancestor(\task{t}) = \bigcup_{i\ge1}{pred^i(\task{t})}$,
where $pred^i(\task{t}) = \bigcup_{\task{s} \in pred^{i-1}(\task{t})}{pred(\task{s})}$ an
d\\
$pred^1(\task{t}) = pred(\task{t})$.
\end{APPdefinition}

\begin{APPdefinition} \label{theorem4_9}
$succ(\task{t})$ is the immediate successor set of task $\task{t}$, formally $succ(\task{t}) = \{\task{c} | \task{t} \in pred(\task{c})\}.$
\end{APPdefinition}

\begin{APPdefinition} \label{theorem4_10}
$descendant(\task{t})$ is a set of descendants of $\task{t}$, formally $descendant(\task{t}) = \bigcup_{i\ge1}{succ^i(\task{t})}$,
where $succ^i(\task{t}) = \bigcup_{\task{s} \in succ^{i-1}(\task{t})}{succ(\task{s})}$ and
$succ^1(\task{t}) = succ(\task{t})$.
\end{APPdefinition}

\begin{APPdefinition} \label{theorem4_11}
$path(\task{s},\task{d}) = descendant(\task{s}) \bigcap ancestor(\task{d}).$
\end{APPdefinition}

Now we will examine the effect of the early finish time of an ancestor task to the finish time of the target task.
In theorem~\ref{theorem4_5}, we examine the relation between the release time and the finish time which
also explains the finish times between predecessor tasks and the target task.
We will extend it to the effect between an ancestor task and the target task.
If we apply theorem~\ref{theorem4_5} repeatedly, we can compute how much the maximum finish time reduction is inherited
from an ancestor task to the target task.

\begin{APPtheorem} \label{theorem4_12}
For $\task{a} \in ancestor(\task{t})$, if $\forall_{\finish{a}}{(\tau_a^{finish} \le \maxF{a}-\Delta_{\task{a}}+pe(\task{a},\Delta_{\task{a}}))}$
then always $\finish{t} \le \maxF{t} - \Delta_{\task{t}} + pe(\task{t},\Delta_{\task{t}})$, 
where $\Delta_{\task{m}} = \max_{\task{i} \in pred(\task{m})}{(\maxF{i})} - \max_{\task{i} \in pred(\task{m})}{(\maxF{i}-\Delta_{\task{i}}+pe(\task{i},\Delta_{\task{i}}))}$.
Note that $\Delta_{\task{m}} = 0$ if there is no predecessor.
\end{APPtheorem}

\begin{proof}
We will prove it by induction.
First, $\finish{a} \le \maxF{a}-\Delta_{\task{a}}+pe(\task{a},\Delta_{\task{a}})$.
Assume that $\finish{i} \le \maxF{i}-\Delta_{\task{i}}+pe(\task{i},\Delta_{\task{i}})$.
For $\task{m} \in succ(\task{i})$,
$\release{m}
= \max_{\task{i} \in pred(\task{m})}{\finish{i}}
\le \max_{\task{i} \in pred(\task{m})}{\maxF{i}-\Delta_{\task{i}}+pe(\task{i},\Delta_{\task{i}})}$.
Let $\Delta_{\task{m}}$ = $\max_{\task{i} \in pred(\task{m})}{(\maxF{i})}- \max_{\task{i} \in pred(\task{m})}{(\maxF{i}-\Delta_{\task{i}}+pe(\task{i},\Delta_{\task{i}}))}$.\\

Since $\maxR{m} = \max_{\task{i} \in pred(\task{m})}{\maxF{i}}$,

$$\release{m} = \max_{\task{i} \in pred(\task{m})}{\finish{i}} \le \max_{\task{i} \in pred(\task{m})}{\maxF{i}-\Delta_{\task{i}}+pe(\task{i},\Delta_{\task{i}})}$$
$$= \max_{\task{i} \in pred(\task{m})}{(\maxF{i})} - \Delta_{\task{m}} = \maxR{m} - \Delta_{\task{m}}.$$

By theorem ~\ref{theorem4_5}, $\finish{m} \le \maxF{m} - \Delta_{\task{m}}+pe(\task{m},\Delta_{\task{m}})$
since $\release{m} \le \maxR{m} - \Delta_{\task{m}}$. Q.E.D.
\end{proof}

Theorem~\ref{theorem4_12} tells us that the contribution of the maximum finish time reduction
of an ancestor task diminishes as it propagates to the child tasks.
To utilize theorem~\ref{theorem4_12},
we define the new reduced finish time $\hat \tau_t^{maxF}$ by reducing the finish time of a task $\task{a}$ by $\Delta$ as following:

\begin{APPdefinition} \label{theorem4_13}
$\hat \tau^{maxF}_{t,\Delta_{\tau_a} \leftarrow \Delta}
= \tau_t^{maxF} - \Delta_{\tau_t} + pe(\tau_t,\Delta_{\tau_t})$,
where $\Delta_{\tau_m} = \max_{\tau_i \in pred(\tau_m)}{\tau_i^{maxF}} - \max_{\tau_i \in pred(\tau_m)}
{\hat \tau^{maxF}_{i,\Delta_{\tau_a} \leftarrow \Delta}}$ if $\Delta_{\tau_a}$ is $\Delta$.
\end{APPdefinition}

By definition~\ref{theorem4_13} and theorem~\ref{theorem4_12},
$\tau_t^{finish} \le \hat \tau_t^{maxF}$ 
if $\forall \tau_a^{finish}, \tau_a^{finish} \le \hat \tau_a^{maxF}$. 
To avoid the duplicate preemptions, we define a common preemptor set as $cp(\tau_a,\tau_t)$.

\begin{APPdefinition} \label{theorem4_14}
For $\tau_a \in ancestor(\tau_t)$,
$cp(\tau_a,\tau_t) = \{\tau_s|\tau_s^{pri}>\max(\tau_a^{pri},\tau_t^{pri}), \tau_s^{maxF} \ge \tau_t^{minR}, \tau_s^{minR} \le \tau_a^{maxF}\}$.
\end{APPdefinition}

If a preemption task can preempt both $\tau_a$ and $\tau_t$ then $\tau_t^{finish}$ is larger
when it preempts $\tau_t$ rather than $\tau_a$.

\begin{APPtheorem} \label{theorem4_15}
If a common preemptor $\tau_p$ can preempt ancestor task $\tau_a$ and target task $\tau_t$ completely,
then always it produces no smaller finish time of $\tau_t$ when $\tau_p$ preempts $\tau_t$ than $\tau_a$.
\end{APPtheorem}

\begin{proof}
The reduced maximum finish time of $\tau_t$ is
$\hat \tau^{maxF}_{t,\Delta_{\tau_a} \leftarrow \tau_p^{WCET}}$ when $\tau_p$
does not preempt $\tau_a$ but preempts $\tau_t$.
It is $\hat \tau^{maxF}_{t,\Delta_{\tau_t} \leftarrow \tau_p^{WCET}}$
when $\tau_p$ does not preempt $\tau_t$ but preempts $\tau_a$.
Since $\tau_a$ is an ancestor of $\tau_t$,
$\hat \tau^{maxF}_{t,\Delta_{\tau_a} \leftarrow \tau_p^{WCET}}
\ge \hat \tau^{maxF}_{t,\Delta_{\tau_t} \leftarrow \tau_p^{WCET}}$.
Hence, it produces no smaller finish time to preempt $\tau_t$ than $\tau_a$. Q.E.D.
\ignore{
Assume that $\delta$ indicates the overlapped execution time
of a common preemption task $\tau_p$ with task $\tau_t$.
Note that it is assumed that $\tau_t$ is preempted by $\tau_p$ for $\delta \le \tau_p^{WCET}$
while $\tau_a$ can be preempted for $\tau_p^{WCET}$ if it is preempted by $\tau_p$.
Let $\tau_{maxF}$ denote the maximum finish time when the common preemption
task preempts both $\tau_a$ and $\tau_t$ simultaneously.
}
\end{proof}

Let us consider a case that a preemption task $\tau_p$ can preempt $\tau_a$ and $\tau_t$, but
$\tau_t$ partially.
Even in this case, $\tau_p$ can always preempt $\tau_a$ fully
since $\tau_p^{finish} \ge \tau_t^{release} \ge \tau_a^{finish}$.
Then we may need to compare
$\hat \tau^{maxF}_{t,\Delta_{\tau_a} \leftarrow \tau_p^{WCET}}$ with
$\hat \tau^{maxF}_{t,\Delta_{\tau_t} \leftarrow maxP(\tau_t,\tau_p)}$ where
$maxP(\tau_t,\tau_p)$ indicates the maximum partial preemption time of $\tau_t$ by $\tau_p$.
Even in this case, surprisingly, preempting $\tau_t$ always provides
larger finish time than the partial preemption case, which is stated by theorem~\ref{theorem4_16}.

\ignore{
Note that $maxP(\tau_t,\tau_p) \le \tau_p^{maxF}-\tau_t^{minR}$.
If $\hat \tau^{maxF}_{t,\Delta_{\tau_a} \leftarrow \tau_p^{WCET}} >
\hat \tau^{maxF}_{t,\Delta_{\tau_t} \leftarrow maxP(\tau_t,\tau_p)}$, then
it provides larger finish time to preempt $\tau_t$; otherwise, it gives larger finish time
to preempt $\tau_a$. Therefore, for the partial preemption case, the optimization algorithm
compares the reduced finish time and chooses the worst case.
[HO: I am not sure whether the supporting of the partial preemption is possible since
too many combinations should be tested whether $\tau_p$ preempts $\tau_a$ or $\tau_t$ if
there are many $\tau_p$. As I know, Junchul's algorithm does not consider the partial
preemption. Can you check it?]
}

\begin{APPtheorem} \label{theorem4_16}
If a common preemptor $\tau_p$ can preempt an ancestor task $\tau_a$ fully and the target task $\tau_t$ partially,
it always produces no smaller finish time bound of $\tau_t$ when $\tau_p$ preempts $\tau_t$ rather than $\tau_a$.
\end{APPtheorem}

\begin{proof}
Assume that $\delta$ indicates the smallest partial preemption time between $\tau_p$ and $\tau_t$.
So, $\delta = \tau_p^{maxF} - \tau_t^{maxR} \le \tau_p^{WCET}$.

\begin{enumerate}
\item If $\tau_p$ preempts $\tau_a$ then $\tau_t$ finishes earlier than $\tau_t^{maxF} - \delta$.
\item If $\tau_p$ preempts $\tau_t$ then assume that $\tau_t^{release} \le \tau_t^{maxR} - \epsilon$.
Then $\tau_t^{finish} \le \tau_t^{maxF} - \epsilon + pe(\tau_t,\epsilon)$.

\begin{enumerate}
\item If $\epsilon \le \tau_p^{WCET} - \delta$ then $pe(\tau_t,\epsilon) = \epsilon$, and
$\tau_t^{maxF} - \epsilon + pe(\tau_t,\epsilon) = \tau_t^{maxF}$. \label{theorem4_16_1}
\item If $\tau_p^{WCET} - \delta < \epsilon$ then $pe(\tau_t,\epsilon) \ge \tau_p^{WCET} - \delta$.
Since $\epsilon \le \tau_p^{WCET}$ by theorem~\ref{theorem4_12},
$\tau_t^{maxF} - \epsilon + pe(\tau_t,\epsilon) \ge \tau_t^{maxF} - \tau_p^{WCET} + \tau_p^{WCET} - \delta
= \tau_t^{maxF} - \delta$. \label{theorem4_16_2}
\end{enumerate}

\end{enumerate}

Hence, in both cases \ref{theorem4_16_1} and \ref{theorem4_16_2}, the preemption of $\tau_t$ provides no smaller maximum finish time
of $\tau_t$.
\ignore{
The reduced maximum finish time of $\tau_t$ is
$\hat \tau^{maxF}_{t,\Delta_{\tau_a} \leftarrow \tau_p^{WCET}}$ when $\tau_p$
does not preempt $\tau_a$ but preempts $\tau_t$.
It is $\hat \tau^{maxF}_{t,\Delta_{\tau_t} \leftarrow \tau_p^{WCET}}$
when $\tau_p$ does not preempt $\tau_t$ but preempts $\tau_a$.
}
\end{proof}

Now, we consider multiple ancestor tasks which have multiple common preemption tasks.
We define a unique common preemptor set, $ucp(\tau_a,\tau_p)$, to indicate
the closest ancestor with the common preemption task as following.

\begin{APPdefinition} \label{theorem4_17}
$$ucp(\tau_a,\tau_t) = cp(\tau_a,\tau_t) - \bigcup_{\tau_m \in path(\tau_a,\tau_t)}{cp(\tau_m, \tau_t)}.$$
\end{APPdefinition}

Then multiple ancestors may have different common preemption tasks. To consider
the cascaded scheduling effect as a whole, we redefine $\Delta_{\tau_a}$.
We introduce $\delta_{\tau_a}$ which denotes the reduced time by moving common preemption tasks
from $\tau_a$ to $\tau_t$. $\delta_{\tau_a}$ is statically computed. $\Delta_{\tau_a}$ indicates
a variable used in definition~\ref{theorem4_13}. The finish time of $\tau_m$ is contributed by the early finish
time of ancestor tasks of $\tau_m$ and $\delta_{\tau_m}$. The total $\Delta_m$ is defined as following:

\begin{APPdefinition} \label{theorem4_18}
$\Delta_{\tau_m} = \max_{\tau_i \in pred(\tau_m)}{\tau_i^{maxF}} - \max_{\tau_i \in pred(\tau_m)}{\hat \tau_i^{maxF}} + \delta_{\tau_m}$,
where $\delta_{\tau_m} = \sum_{\tau_s \in ucp(\tau_m,\tau_t)}{\min(\tau_s^{WCET},\tau_s^{maxF}-\tau_t^{maxR})}$.
\end{APPdefinition}

Finally, we define $pe(\tau_t,\Delta)$ which indicates the occupied time from $\tau_t^{maxR}-\Delta$ to $\tau_t^{maxR}$. Although $pe(\tau_t,\Delta)$ shows the exact time, it is varying at run time. Therefore,
we need to compute its bound. Although we use its bound, the previous theorems hold.
Since $pe(\tau_t,\Delta) \le \min(\Delta, maxpe(\tau_t)[\tau_t^{maxR}-\Delta,\tau_t^{maxR}])$
where $maxpe(\tau_t)[x,y] =
\sum_{\tau_s^{pri} > \tau_t^{pri}, \tau_s^{maxF} \ge x, \tau_s^{minS} \le y}{\min(\tau_s^{WCET}, y-\tau_s^{minS}, \tau_s^{maxF} - x)}$,
$\min(\Delta, maxpe(\tau_t)[\tau_t^{maxR}-\Delta,\tau_t^{maxR}])$
is used as the bound of $pe(\tau_t,\Delta)$.

\begin{APPdefinition} \label{theorem4_19}
$pe(\tau_t,\Delta) = \min(\Delta, maxpe(\tau_t)[\tau_t^{maxR}-\Delta,\tau_t^{maxR}])$,\\
where $maxpe(\tau_t)[x,y] = \sum_{\tau_s^{pri} > \tau_t^{pri}, \tau_s^{maxF} \ge x, \tau_s^{minS} \le y}{\min(\tau_s^{WCET}, y-\tau_s^{minS}, \tau_s^{maxF} - x)}$.
\end{APPdefinition}

\begin{APPtheorem} \label{theorem4_20}
Algorithm \ref{alg:commonpreemption} that removes the duplicate preemptions from the ancestors in the critical path preserves the conservativeness of the schedule time bound.
\end{APPtheorem}

\begin{proof}
By replacing variables in definition~\ref{theorem4_13} by definitions~\ref{theorem4_18} and~\ref{theorem4_19},
we can compute the tighter time bound with moving the common preemption tasks from the ancestors. The conservativeness is guaranteed by Theorems \ref{theorem4_15} and \ref{theorem4_16} since preempting $\task{t}$ rather than ancestor $\task{a}$ produces a larger time bound. The difference between the proof and the actual technique is that \emph{duplicate preemption elimination} technique traverses only the critical path, not all ancestors. Even though common preemptions of the ancestors not in the critical path remain, our technique is still conservative. Q.E.D.
\end{proof}

\end{document}